\documentclass[12pt,oneside]{book}
\usepackage{amsmath}
\usepackage{amsthm}
\usepackage{graphics}
\usepackage{epsfig,amssymb,latexsym,verbatim}
\usepackage{graphicx}
\usepackage{multirow}
\usepackage{multicol}
\usepackage{float}
\usepackage{booktabs}
\usepackage{fancyhdr}
\usepackage{enumerate,verbatim}
\usepackage[sort&compress]{natbib}
\usepackage{rotating}
\usepackage{hyperref}
\usepackage{subfigure}
\usepackage{ulem}
\usepackage{url}
\usepackage{threeparttable}
\usepackage{xr}
\usepackage{multirow}
\usepackage{multicol}
\usepackage{wrapfig}
\usepackage{lscape}
\usepackage{rotating}
\usepackage{epstopdf}
\usepackage{subfigure}
\usepackage{amsfonts}
\usepackage[usenames,dvipsnames]{color}
\usepackage{times,color}
\usepackage{amsbsy}
\usepackage{amsthm}
\usepackage[sf,small,center, compact]{titlesec}

\def\bs{\boldsymbol}

\definecolor{red}{rgb}{1,0,0}
\definecolor{green}{rgb}{0,1,0}
\definecolor{blue}{rgb}{0,0,1}

\definecolor{lxy}{RGB}{180,0,180}

\floatstyle{ruled}
\newfloat{algorithm}{tbp}{loa}
\floatname{algorithm}{Algorithm}
\def\bs{\boldsymbol}

\newcommand{\thetheorem}{{\thesection. \arabic{theorem}}}
\newcommand{\thelemma}{{\thesection. \arabic{lemma}}}
\newcommand{\theproposition}{{\thesection. \arabic{proposition}}}
\newcommand{\thecorollary}{{\thesection. \arabic{corollary}}}
\newtheorem{theorem}{{\sc Theorem}}
\newtheorem{lemma}{{\sc Lemma}}
\newtheorem{corollary}{{\sc Corollary}}

\begin{document}
\renewcommand{\baselinestretch}{1.2}
\markboth{\hfill{\footnotesize\rm Guannan Wang and Jue Wang}\hfill}
{\hfill {\footnotesize\rm On Selection of Semiparametric Spatial Regression Models} \hfill}
\renewcommand{\thefootnote}{}
$\ $\par \fontsize{10.95}{14pt plus.8pt minus .6pt}\selectfont
\vspace{0.8pc} \centerline{\large\bf On Selection of Semiparametric Spatial Regression Models}
\vspace{.4cm} \centerline{Guannan Wang$^{a}$ and Jue Wang$^{b}$
\footnote{\emph{Address for correspondence}: Guannan Wang, Department of Mathematics, College of William \& Mary, Williamsburg, VA, USA. Email: gwang01@wm.edu}} \vspace{.4cm} \centerline{\it $^{a}$College of William \& Mary and $^{b}$Iowa State University} \vspace{.55cm}
\fontsize{9}{11.5pt plus.8pt minus .6pt}\selectfont

\begin{quotation}
\noindent \textit{Abstract:} In this paper, we focus on the variable selection techniques for a class of semiparametric spatial regression models which allow one to study the effects of explanatory variables in the presence of the spatial information. The spatial smoothing problem in the nonparametric part is tackled by means of bivariate splines over triangulation, which is able to deal efficiently with data distributed over irregularly shaped regions. In addition, we develop a unified procedure for variable selection to identify significant covariates under a double penalization framework, and we show that the penalized estimators enjoy the ``oracle" property. The proposed method can simultaneously identify non-zero spatially distributed covariates and solve the problem of ``leakage" across complex domains of the functional spatial component. To estimate the standard deviations of the proposed estimators for the coefficients, a sandwich formula is developed as well. In the end, Monte Carlo simulation examples and a real data example are provided to illustrate the proposed methodology. All technical proofs are given in the appendixes.

\vspace{9pt} \noindent \textit{Key words and phrases:} Bivariate splines, Partially linear models, Penalized regression, Semiparametric regression, Spatial data.
\end{quotation}

\fontsize{10.95}{14pt plus.8pt minus .6pt}\selectfont

\thispagestyle{empty}

\setcounter{chapter}{1}
\setcounter{equation}{0}
\noindent \textbf{1. Introduction} \vskip 0.1in
\label{SEC:introduction}

In many economic and geographic data studies, we may have spatially-referenced covariates providing information regarding the spatial distribution which impact the response variable of interest. Meanwhile, many other explanatory variables could be introduced to the model to help explain the response variable. For example, the mortality dataset described in Section 6 consists of aggregated data from each of 3,037 counties in the United States; see Figure \ref{FIG:Mortality}. The explanatory variables contain the county level social, economic and ethnic information that could affect the mortality rate.

\begin{figure}[t]
	\begin{center}
		\includegraphics[scale=0.25]{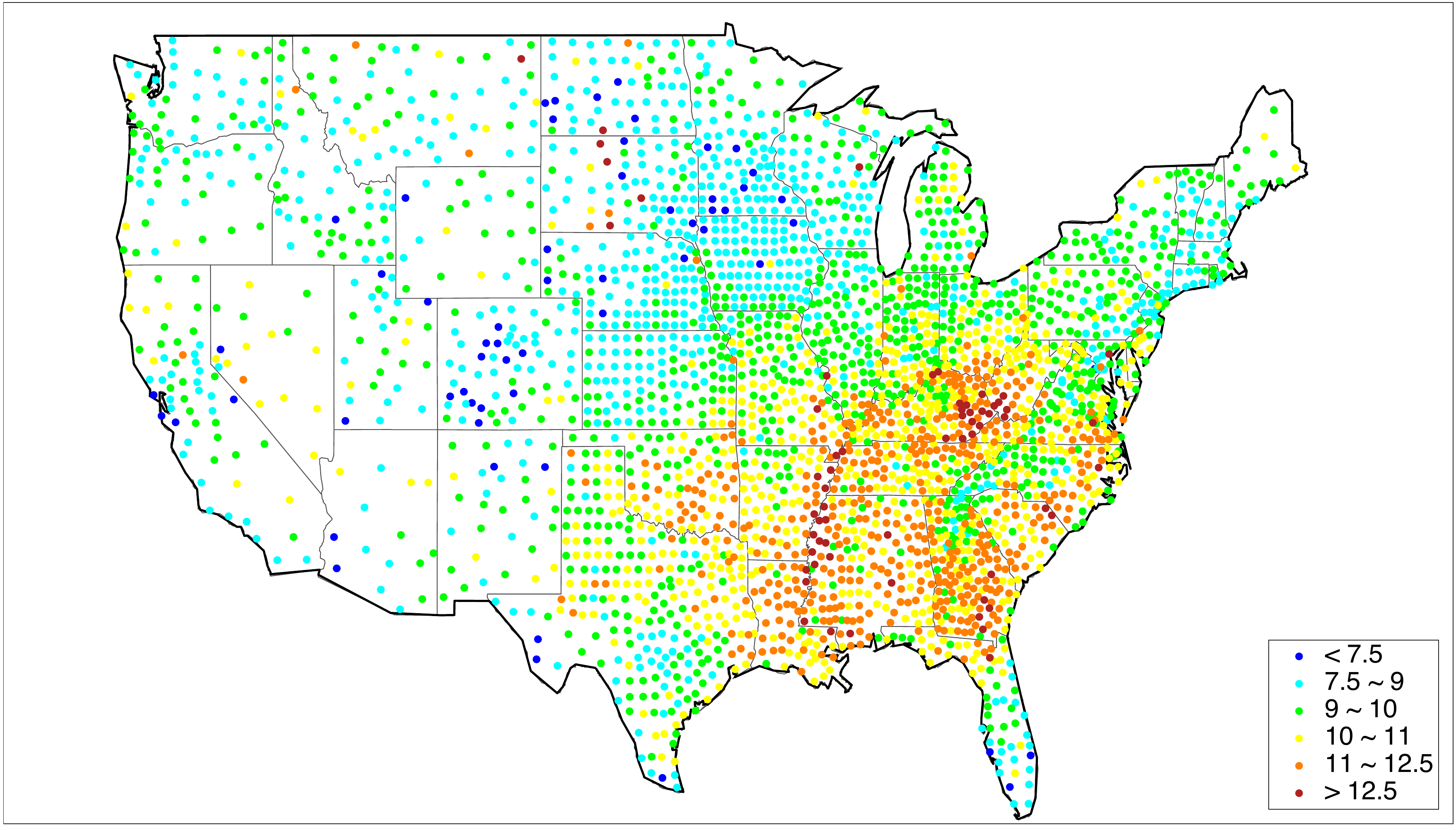}
		\caption{Mortality rate from 3,037 counties in the U.S.\label{FIG:Mortality}}
	\end{center}
\end{figure}

To incorporate the spatial information into the regression, there are mainly two kinds of modeling approaches. The first approach adds spatial correlation structure (or weights) to a regression modeling to include spatial information, for example, \cite{Leung:Cooley:14} provided a through comparison of the predictive ability of a traditional geostatistical model with that of a non-traditional Gaussian process model; \citep{Lee:2004,Hoshino:18,LeSage:Pace:09,Wall:04} studied the spatial autoregressive (SAR) model and the conditional autoregressive (CAR) model; and \cite{Nandy:Lim:Maiti:17} considered the spatially weighted regression (SWR) method. A second approach is based on some smoothing techniques, for example, kernel, wavelet or spline smoothing, which uses a deterministic smooth bivariate function to describe the variations and connections among values at different locations; see, for example, \cite{Gheriballah:Laksaci:Rouane:10}, \cite{Ramsay:02}, \cite{Wood:03}, \cite{strand2006wavelet}, \cite{Sangalli:Ramsay:Ramsay:13} and \cite{Lai:Wang:13}. In this paper, we take the second approach. We focus on the partially linear spatial model (PLSM) containing both linear components and a nonparametric component with spatial information involved for data distributed over a two-dimensional (2-D) domain.

Suppose there are $n$ location points ranging over a bounded domain $\Omega \subseteq \mathbb{R}^2$ of arbitrary shape. For the $i$th location point, we observe $(\mathbf{Z}_{i},\mathbf{X}_{i},Y_{i})$, where $\mathbf{Z}_{i}=(Z_{i1},\ldots,Z_{ip})^{\top}$ is a $p$-dimensional vector of the predictor variables. For example, in the mortality data analysis, the vector $\mathbf{Z}$ contains socioeconomic and race/ethnicity information such as Gini coefficient, social affluence and proportion of African-American, etc. Variable $\mathbf{X}_{i}=(X_{i1},X_{i2})^{\top}$ represents the location (longitude and latitude), and $Y_i$ is the response variable of interest. We consider the following semiparametric regression model
\begin{equation}
Y_{i}=\mathbf{Z}_i^{\top} \bs{\beta}+\alpha\left(\mathbf{X}_{i}\right)+\epsilon _{i}, \quad i=1,\ldots ,n,
\label{model}
\end{equation}
where $\bs{\beta}=(\beta_{1},\ldots,\beta_{p})^{\top}$ are unknown parameters and $\alpha(\cdot)$
is some unknown but smooth bivariate function. When $\alpha(\cdot)$ is a univariate function, this model is the traditional partially linear model (PLM), and its estimation and theoretical properties have been well studied; see, for example, 
\cite{Huang:Zhang:Zhou:07}, \cite{He:Tang:Zuo:14} and \cite{Brown:Levine:Wang:16}. Following the nonparametric smoothing approach, the random noises, $\epsilon_{i}$'s, are assumed to be i.i.d with $E\left( \epsilon_{i}\right)=0$ and $\mathrm{Var}\left(\epsilon _{i}\right) =\sigma^2$, and each $\epsilon_{i}$ is independent of $\mathbf{X}_{i}$ and $\mathbf{Z}_{i}$.

While it may be desirable to include many predictors in the model, there are practical constraints that limit the feasibility of such an approach. For example, one problem that arises when analyzing high dimensional data is the computation efficiency. Increasing model sparsity enforces a lower dimensional model structure; thus the model can be estimated more efficiently. In addition, it makes inference more tractable, models easier to interpret, and leads to more robustness against noise.

Variable selection has been well studied in the partially linear model (PLM) literature with univariate functional component $\alpha(\cdot)$; see \cite{Bunea:Wegkamp:04,Liang:Li:09,Xie:Huang:09} and among others. When $\mathbf{X}$ is bivariate or multivariate, existing variable selection methods have been largely limited to the additive model (AM) or partially linear additive model (PLAM) which approximates the surface by a combination of an additive collection of univariate functions; see, for example, \cite{Ma:Yang:11,Ma:Song:Wang:13,Lian:Liang:Wang:14,Liu:Wang:Liang:11,Wang:Liu:Liang:Carroll:11,Lian:12}. However, many spatial studies showed that the additive assumption is questionable in the two-dimensional (2-D) domain.

Traditional bivariate smoothing tools require that data distributed nicely on a rectangular domain. However, the domain over which variables of interest are defined in many of the spatial data is often found to be irregular and complicated. It is challenging to achieve variable selection for irregularly spaced spatial data distributed over complex domains, and the challenges include (i) how to identify those important covariates in $\mathbf{Z}$, (ii) how to estimate the bivariate function $\alpha(\cdot)$ ranging over some irregular 2-D domains, (iii) how to deal with unevenly distributed data with observations dense at some locations while sparse at others, (iv) how to make the estimation and selection both computationally efficient and theoretically reliable.

To approximate the bivariate $\alpha(\cdot)$, many of the approaches involve tensor product estimation. However, the structure of tensor products is most useful when the data are observed in a rectangular domain, and is undesirable when data are located in spatial domains with complex boundary features such as the domain of the US; see Figure \ref{FIG:Mortality}. Many conventional smoothing tools (kernel smoothing, wavelet smoothing and tensor product splines) suffer from the problem of ``leakage'' across the complex domains, which refers to the poor estimation over difficult regions by smoothing inappropriately across boundary features; see more discussions in \cite{Ramsay:02} and \cite{Wood:Bravington:Hedley:08}.

To this end, we aim to address questions (i)-(iv). To deal with the irregular domain problem in (ii), \cite{Sangalli:Ramsay:Ramsay:13} proposed to use finite element analysis, a method mainly developed and used to solve partial differential equations \cite{Wood:Bravington:Hedley:08} proposed the soap film smoothing method. Although their method is useful in many practical applications, the theoretical properties of the smoothing were not studied in their paper. In this paper, we will approximate the nonparametric function $\alpha(\cdot)$ using bivariate splines, i.e., smooth piecewise polynomial functions, over triangulations \citep{Lai:Schumaker:07}. This method solves the problem of ``leakage'' across the complex domains, and it does not require constructing finite elements or locally supported basis functions. It is also computationally efficient, and provides excellent approximation properties \citep{Lai:Wang:13}, thus, it can handle part of the challenges mentioned in (iv).

To further meet the challenges in (i), (iii) and (iv), we incorporate a variable selection mechanism into the PLSM and propose a double penalized least squares approach based on  bivariate spline approximation over the spatial domain. Roughness penalty based on the second-order derivatives is employed to regularize the spline fit, and shrinkage penalty on parametric components is used to achieve the variable selection. When we have regions of sparse data, penalized splines provide a more convenient tool for data fitting than the unpenalized splines. We demonstrate that the estimator possesses the ``oracle" property in the sense that it is as efficient as the estimator when the true model is known prior to statistical analysis. We also propose a coordinate descent based algorithm to perform the variable selection efficiently.

The rest of this article is organized as follows. In Section 2, we first introduce the triangulations and bivariate spline spaces, then we propose the double-penalized least squares method for joint variable selection and model estimation, and define the penalized estimator $(\widehat{\bs{\beta}}, \widehat{\alpha})$. In Section 3 , we further study the asymptotic properties of the estimator $\widehat{\bs{\beta}}$ including the consistency and ``oracle" property, as well as the rate of convergence of $\widehat{\alpha}$. In Section 4, we discuss some implementation details such as how to select the tuning parameters.  Sections 5 and 6 present simulation results and a real data application of the mortality data. Section 7 concludes the paper with a discussion of related issues. Technical details are provided in the appendixes.

\setcounter{chapter}{2} \label{SEC:method} \renewcommand{\theproposition}{2.\arabic{proposition}}
\renewcommand{\thetable}{2.\arabic{table}} \setcounter{table}{0} 
\renewcommand{\thefigure}{2.\arabic{figure}} \setcounter{figure}{0}  \setcounter{equation}{0} \setcounter{lemma}{0} \setcounter{theorem}{0} \setcounter{proposition}{0} \setcounter{corollary}{0}
\vskip .12in \noindent \textbf{2. Methodology} \vskip 0.10in

We approximate the function $\alpha(\cdot)$ by bivariate splines defined over a 2D triangulated domain. In the following, we use $\tau$ to denote a triangle which is a convex hull of three points not located in one line. A collection $\triangle=\{\tau_1,\ldots,\tau_K\}$ of $K$ triangles is called a triangulation of $\Omega=\cup_{k=1}^{K}\tau_k$ provided that if a pair of triangles in $\triangle$ intersect, then their intersection is either a common vertex or a common edge. See Figures \ref{FIG:eg1-2} and \ref{FIG:Mortality-tri} for some examples of triangulations.

Various packages have been developed for triangulation; see for example, the ``Delaunay'' algorithm (MATLAB program \textit{delaunay.m} or MATHEMATICA function \textit{DelaunayTriangulation}), the ``Triangle" (\url{http://www.cs.cmu.edu/~quake/triangle.html}) by \cite{Shewchuk:96}, and the ``DistMesh'' (\url{http://persson.berkeley.edu/distmesh}).

\vskip .10in \noindent \textbf{2.1. Penalized spline estimators} \vskip .10in

For a nonnegative integer $r$, let $\mathbb{C}^r(\Omega)$ be the collection of all $r$-th continuously differentiable functions over $\Omega$. Given a triangulation $\triangle$, let $\mathbb{S}_{d}^{r}(\triangle)=\{s\in\mathbb{C}^{r}(\Omega):s|_{\tau}\in \mathbb{P}_{d}(\tau), \tau \in \triangle\}$ be a spline space of degree $d$ and smoothness $r$ over triangulation $\triangle $, where $s|_{\tau}$ is the polynomial piece of spline $s$ restricted on triangle $\tau$, and  $\mathbb{P}_{d}$ is the space of all polynomials of degree less than or equal to $d$. It has been proved in \cite{Lai:Schumaker:07} that for a fixed smoothness $r\geq1$, the spline space $\mathbb{S}_{3r+2}^{r}(\triangle)$ achieves the optimal rate of convergence for noise-free datasets, thus, for notation simplicity, we let $\mathbb{S}=\mathbb{S}_{3r+2}^{r}(\triangle)$. Given a $\lambda >0$ and $\{(\mathbf{Z}_i,\mathbf{X}_i,Y_i)\}_{i=1}^{n}$, we consider the following minimization problem:
\begin{eqnarray}
\min_{s\in \mathbb{S}}\frac{1}{2}\sum_{i=1}^{n}\left\{Y_{i}-\mathbf{Z}_{i}^{\top} \bs{\beta}-s\left(
\mathbf{X}_{i}\right)\right\} ^{2}+\frac{1}{2}\lambda
\mathcal{E}(s), \label{DEF:minimization}
\end{eqnarray}
where
\begin{equation*}
\mathcal{E}(s)=\int_{\Omega} \left\{\left(\frac{\partial^{2}}{\partial x_{1}^2} s\right)^2+2\left(\frac{\partial^{2}}{\partial x_{1} \partial x_{2}} s\right)^2+\left(\frac{\partial^{2}}{\partial x_{2}^2} s\right)^2\right\}dx_{1}dx_{2}.
\label{penalty-2}
\end{equation*}

We use Bernstein basis polynomials to represent the bivariate splines. Let $\{B_{\xi }\}_{\xi \in \mathcal{K}}$ be the set of degree-$d$ bivariate Bernstein basis polynomials for $\mathbb{S}$ constructed in \cite{Lai:Schumaker:07}, where $\mathcal{K}$ stands for an index set of $K$ Bernstein basis polynomials. Then we can write the function $s(\mathbf{x})=\sum_{\xi \in \mathcal{K}} B_{\xi}(\mathbf{x})\gamma_{\xi}
=\mathbf{B}(\mathbf{x})^{\top}\bs{\gamma}$, where $\bs{\gamma}^{\top} =(\gamma_{\xi},\xi \in \mathcal{K})$
is the spline coefficient vector. To meet the smoothness requirement of the splines, we need to impose some constraints on the spline coefficients. Denote $\mathbf{H}$ the constraint matrix  on the coefficients $\bs{\gamma}$, which depends on $r$ and the structure of the triangulation and enforces smoothness across shared edges of triangles. A simple example of $\mathbf{H}$ is given in \cite{Zhou:Pan:14}. The smoothness conditions are linear, and can be written as $\mathbf{H}\bs{\gamma}=\mathbf{0}$.

Let $\mathbf{Y} = (Y_1,\ldots,Y_n)^{\top}$ be the vector of $n$ observations of the response variable. Denote by $\mathbf{X} _{n\times 2}= \{(X_{i1}, X_{i2})\}_{i=1}^{n}$ the design matrix of the locations and  $\mathbf{Z} _{n\times p}= \{(Z_{i1}, \ldots,Z_{ip})\}_{i=1}^{n}$ the collection of all covariates. Denote by $\mathbf{B}$ the $n\times K$ evaluation matrix of Bernstein basis polynomials whose $i$-th row is given by $\mathbf{B}_{i}^{\top}=\{B_{\xi}(\mathbf{X}_{i}), \xi\in \mathcal{K}\}$. Then the minimization problem in (\ref{DEF:minimization}) reduces to
\begin{equation}
\min_{\bs{\beta},\bs{\gamma}} \frac{1}{2}\left\{\|\mathbf{Y}-\mathbf{Z}\bs{\beta}-\mathbf{B}\bs{\gamma}\|^{2}+\lambda\bs{\gamma}^{\top}\mathbf{P}\bs{\gamma}\right\}~~ \mathrm{subject~to} ~~ \mathbf{H}\bs{\gamma}=\mathbf{0},
\label{EQ:minimization}
\end{equation}
where $\mathbf{P}$ is the diagonally block penalty matrix satisfying that $\bs{\gamma}^{\top}\mathbf{P}\bs{\gamma}=\mathcal{E}(\mathbf{B}\bs{\gamma})$.

To solve the constrained minimization problem (\ref{EQ:minimization}), we first remove the constraint via a QR decomposition of the transpose of matrix $\mathbf{H}$ and convert the problem to a conventional penalized regression problem without any restriction. More specifically, we assume
$\mathbf{H}^{\top}=\mathbf{Q}\mathbf{R}=\left(\mathbf{Q}_1 ~\mathbf{Q}_2\right)
\binom{\mathbf{R}_{1}}{\mathbf{R}_{2}}$,
where $\mathbf{Q}$ is an orthogonal matrix and $\mathbf{R}$ is an upper triangle matrix; the submatrix $\mathbf{Q}_1$ is the first $r_H$ columns of $\mathbf{Q}$, where $r_H$ is the rank of matrix $\mathbf{H}$, and $\mathbf{R}_2$ is a matrix of zeros. We reparameterize using $\bs{\gamma} = \mathbf{Q}_2\bs{\theta}$ for some $\bs{\theta}$, and it has been proved in \cite{Wang:Wang:Lai:Gao:18} that after the reparameterization $\mathbf{H}\bs{\gamma}$ is guaranteed to be $\mathbf{0}$. Then the problem (\ref{EQ:minimization}), is now changed to
\begin{equation}
\min_{\bs{\beta},\bs{\theta}} \left\{\frac{1}{2}\|\mathbf{Y}-\mathbf{Z}\bs{\beta}-\mathbf{B}\mathbf{Q}_{2}\bs{\theta}\|^{2}+\frac{\lambda}{2}(\mathbf{Q}_{2}\bs{\theta})^{\top}\mathbf{P}(\mathbf{Q}_{2}\bs{\theta})\right\}.
\label{EQ:minimization1}
\end{equation}

\vskip .12in \noindent \textbf{2.2. Doubly penalized spline estimators} \vskip .10in

Note that for any fixed $\bs{\beta}$, the minimizer of (\ref{EQ:minimization1}) with respect to $\bs{\theta}$ is
\begin{equation}
\bs{\theta}(\bs{\beta};\lambda)=\left\{\mathbf{Q}_{2}^{\top}(\mathbf{B}^{\top} \mathbf{B}+\lambda\mathbf{P})\mathbf{Q}_{2}\right\}^{-1}
\mathbf{Q}_{2}^{\top}\mathbf{B}^{\top} (\mathbf{Y}-\mathbf{Z}\bs{\beta}),
\label{EQ:theta(beta)}
\end{equation}
Replacing $\bs{\theta}$ by $\bs{\theta}(\bs{\beta};\lambda)$ in (\ref{EQ:minimization1}), we define
\begin{align}
L(\bs{\beta}) \equiv L(\bs{\beta};\lambda) &=\frac{1}{2}\|\mathbf{Y}-\mathbf{Z}\bs{\beta}-\mathbf{B}\mathbf{Q}_{2}\bs{\theta}(\bs{\beta};\lambda)\|^{2}
+\frac{\lambda}{2}\{\mathbf{Q}_{2}\bs{\theta}(\bs{\beta};\lambda)\}^{\top}
\mathbf{P}\{\mathbf{Q}_{2}\bs{\theta}(\bs{\beta};\lambda)\}\notag\\
&=\frac{1}{2}(\mathbf{Y}-\mathbf{Z}\bs{\beta})^{\top}
\{\mathbf{I}-\mathbf{H}_{\mathbf{B}}(\lambda)\}(\mathbf{Y}-\mathbf{Z}\bs{\beta}),
\label{DEF:Lbeta}
\end{align}
where
\begin{equation}
\mathbf{H}_{\mathbf{B}}(\lambda)=\mathbf{B}\mathbf{Q}_{2}\left\{\mathbf{Q}_{2}^{\top}(\mathbf{B}^{\top} \mathbf{B}+\lambda\mathbf{P})\mathbf{Q}_{2}\right\}^{-1}
\mathbf{Q}_{2}^{\top}\mathbf{B}^{\top}.
\label{DEF:HB}
\end{equation}

To achieve the simultaneous estimation of the bivariate function $\alpha(\cdot)$ and the selection of important covariates, we propose a double-penalized least squares method via minimizing
\begin{align}
R(\bs{\beta};\lambda_{1},\lambda_{2}) &=L(\bs{\beta};\lambda_{1})+n\sum_{j=1}^{p}p_{\lambda_{2}}(|\beta_{j}|),
\label{DEF:penalized-min}
\end{align}
where $\lambda_{1}$ and $\lambda_{2}$ are tuning parameters. The first penalty term in (\ref{DEF:penalized-min}) penalizes the roughness of the nonparametric fit $\alpha(\cdot)$ and the second penalty is the shrinkage penalty which shrinks small components of the linear estimates to zero. Various penalty functions have been used in the literature of variable selection for regression models. For example, the LASSO penalty, $p_{\lambda_{2}}(|\beta|)=\lambda_{2}|\beta|$, the Adaptive LASSO (ALASSO) penalty in \cite{Zou:06} is given by $p_{\lambda_{2}}(\beta)=\lambda_{2} w^{*}|\beta|$ for a known data-driven weight $w^{*}$, and the smoothly clipped absolute deviation (SCAD) penalty in \cite{Fan:Li:01}. In this paper, we consider the SCAD penalty defined below:
\begin{equation*}
p_{\lambda_{2}}'(\beta)=\lambda_{2}\left\{I(\beta\leq\lambda_{2})
+\frac{(a\lambda_{2}-\beta)_{+}}{(a-1)\lambda_{2}}I(\beta>\lambda_{2})\right\},
\end{equation*}
for some $a>2$ and $\beta>0$ and $a=3.7$ is used as suggested in \cite{Fan:Li:01}.

The SCAD-penalized estimator of the coefficient $\bs{\beta}$ is then defined as follows:
$
\widehat{\bs{\beta}}=\arg \min_{\bs{\beta}\in \mathbb{R}^{p}}R(\bs{\beta};\lambda_{1},\lambda_{2})
$, 
and the bivariate spline estimator of $\alpha(\mathbf{x})$ is
\begin{equation}
\widehat{\alpha}(\mathbf{x})=\mathbf{B}(\mathbf{x})^{\top}\mathbf{Q}_{2}\left\{\mathbf{Q}_{2}^{\top}(\mathbf{B}^{\top} \mathbf{B}+\lambda_{1}\mathbf{P})\mathbf{Q}_{2}\right\}^{-1}
\mathbf{Q}_{2}^{\top}\mathbf{B}^{\top} (\mathbf{Y}-\mathbf{Z}\widehat{\bs{\beta}}).
\label{DEF:g_hat}
\end{equation}

\setcounter{chapter}{3} \label{SEC:asymptotics} \renewcommand{\thetheorem}{3.\arabic{theorem}}
\renewcommand{\thelemma}{3.\arabic{lemma}}
\renewcommand{\theproposition}{3.\arabic{proposition}}
\renewcommand{\thetable}{3.\arabic{table}} \setcounter{table}{0} 
\renewcommand{\thefigure}{3.\arabic{figure}} \setcounter{figure}{0} 
\setcounter{equation}{0} \setcounter{lemma}{0} \setcounter{theorem}{0}
\setcounter{proposition}{0}\setcounter{corollary}{0}
\vskip .12in \noindent \textbf{3. Asymptotic Results} \vskip 0.10in

In this section, we study the asymptotic properties of the SCAD-penalized partially linear bivariate spline estimator $(\widehat{\bs{\beta}}, \widehat{\alpha})$. We first introduce some notation. For any function $f$ over the closure of domain $\Omega$, denote $\Vert f\Vert _{\infty} =\sup_{\mathbf{x}\in \Omega}|f(\mathbf{x})|$ the supremum norm of function $f$  over $\Omega$, and denote $|f|_{\upsilon,\infty}=\max_{i+j=\upsilon}\left\Vert \frac{\partial^{\upsilon}}{\partial x_{1}^{i} \partial x_{2}^{j}}f(x_{1},x_{2})\right\Vert _{\infty}$
the maximum norm of all the $\upsilon $th order derivatives of $f$ over $\Omega$. Let
\begin{equation}
W^{\ell,\infty }(\Omega)=\left\{f \mathrm{~on~} \Omega:|f|_{k,\infty}<\infty, 0\le k\le \ell \right\}
\label{DEF:Sobolev}
\end{equation}
be the standard Sobolev space. For any $j=1,\ldots, p$, let $z_{j}$ be the coordinate mapping that maps $\mathbf{z}$ to its $j$th component so that $z_{j}(\mathbf{Z}_{i})=Z_{ij}$, and let
\begin{equation}
h_j=\mathrm{argmin}_{h\in L^{2}}\|z_{j}-h\|_{L^{2}}^{2}=\mathrm{argmin}_{h\in L^{2}}E\{Z_{ij}-h(\mathbf{X}_{i})\}^{2}
\label{EQ:h_j}
\end{equation}
be the orthogonal projection of $z_{j}$ onto $L^{2}$.

\vskip .12in \noindent \textbf{3.1. Assumptions} \vskip .10in

Given a triangle $\tau\in \triangle$, let $|\tau|$ be its longest edge length, and $\rho_{\tau}$ be the radius of the largest disk which can be inscribed in $\tau$. Define the shape parameter of $\tau$ as the ratio $\nu_{\tau}=|\tau|/\rho_\tau$. When $\nu_\tau$ is small, the triangle is relatively uniform in the sense that all angles are relatively the same. Denote the size of $\triangle$ by $|\triangle |:=\max \{|\tau|,\tau \in \triangle \}$, i.e., the length of the longest edge of $\triangle$.

Before we state the results, we make the following assumptions:\\
\textbf{Assumption 1.}  The covariates $Z_{ij}$ are bounded uniformly in $i=1, \ldots, n$, $j =1,\ldots, p$.\\
\textbf{Assumption 2.}  The eigenvalues of $E\{(1 ~~\mathbf{Z}_{i}^{\top})^{\top}(1 ~~\mathbf{Z}_{i}^{\top})| \mathbf{X}_{i}\}$ are bounded away from 0.\\
\textbf{Assumption 3.}   The noise $\epsilon$ satisfies that $\lim_{\eta\rightarrow\infty}E\left[\epsilon^{2}I(\epsilon>\eta)\right]
=0$.\\
\textbf{Assumption 4.}  The bivariate functions $h_{j}(\cdot)$, $j=1,\ldots, p$, and the true function in model (\ref{model}), $\alpha(\cdot) \in W^{\ell+1 ,\infty}(\Omega)$, in (\ref{DEF:Sobolev}) for an integer $\ell \ge 2$.\\
\textbf{Assumption 5.}  The joint density of $\mathbf{X}=(X_1,X_2)$ is bounded away from zero and infinity.\\
\textbf{Assumption 6.}  The triangulation $\triangle $ is $\nu$-quasi-uniform, that is, there exists a positive constant $\nu$ such that the triangulation $\triangle$ satisfies  $\nu_{\tau} \leq \nu$,  for all $\tau\in \triangle$.\\
\textbf{Assumption 7.}  The number of the triangles $K$ and the sample size $n$ satisfy that $K=Cn^{\gamma}$ for some constant $C>0$ and $1/(\ell+1)\leq \gamma \leq 1/3$.\\
\textbf{Assumption 8.}  The roughness penalty parameter $\lambda_{1}$ satisfies $\lambda_{1}=o(n^{1/2}K^{-1})$.

Assumptions 1--3 are typical in semiparametric smoothing literature, see for instance, \cite{Huang:Zhang:Zhou:07} and \cite{Wang:Liu:Liang:Carroll:11}. The purpose of Assumption 2 is to ensure that the covariate vector $\mathbf{Z}$ is not multi-collinear. Assumption 4 describes the requirement for the true bivariate function as usually used in the literature of nonparametric or semiparametric estimation; see \cite{Lai:Wang:13}. Assumptions 5--6 require that the partition is quasi-uniform, and suggest that we should not put too few or too many observations in one triangle. Assumption 7 requires that the number of triangles is above some minimum depending upon the degree of the spline, which is similar to the requirement of \cite{Li:Ruppert:08} in the univariate case. Assumption 8 is required to reduce the bias of the bivariate spline approximation through ``under smoothing'' and ``choosing smaller roughness penalty".

\vskip .12in \noindent \textbf{3.2. Sampling properties for the penalized estimators} \vskip .10in
\label{SEC:penalized estimators}

We next show that with a proper choice of $\lambda _{1}$ and $\lambda _{2}$, the penalized
estimator $\widehat{\bs{\beta}}$ has an ``oracle" property. To avoid confusion, let $\bs{\beta}_0$ and $\alpha_{0}$ be the true parameter value and function in model (\ref{model}). Let $q$
be the number of nonzero components of $\bs{\beta}_{0}$.
Let $\bs{\beta}_{0}=(\beta _{10},\cdots ,\beta _{p0})^{\top}=(\bs{\beta}_{10}^{\top},\bs{\beta}_{20}^{\top})^{\top}$, where $\bs{\beta}_{10}$ is assumed to consist of all $q$ nonzero components of $\bs{\beta}_{0}$,
and $\bs{\beta}_{20}=\mathbf{0}$ without loss of generality. Then $\bs{\widehat{\beta}}_{1}$ and $\bs{\widehat{\beta}}_{2}$ are the corresponding estimators. In a similar fashion to $\bs{\beta}$, {we write $\mathbf{Z}=(\mathbf{Z}_{1},\mathbf{Z}_{2})$, and $\widetilde{\mathbf{Z}}=(\widetilde{\mathbf{Z}}_{1},
	\widetilde{\mathbf{Z}}_{2})$, where
	\begin{equation}
	\widetilde{\mathbf{Z}}_{1}=\left\{h_{1}(\mathbf{X}_{i}), \ldots, h_{q}(\mathbf{X}_{i})\right\}_{i=1}^{n},~\widetilde{\mathbf{Z}}_{2}=\left\{h_{q+1}(\mathbf{X}_{i}), \ldots, h_{p}(\mathbf{X}_{i})\right\}_{i=1}^{n}
	\label{DEF:Z1tilde}
	\end{equation}
	with $h_{j}(\cdot)$ defined in (\ref{EQ:h_j}).} Next we denote $a_{n,\lambda_{2}}=\max_{1\leq j\leq p}\{|p_{\lambda _{2}}^{\prime}(|\beta_{j0}|)|,\beta _{j0}\neq 0\}$, $b_{n,\lambda_{2}}=\max_{1\leq j\leq p}\{|p_{\lambda _{2}}^{\prime \prime }(|\beta _{j0}|)|,\beta _{j0}\neq 0\}$.

\begin{theorem}
	\label{THM:ROOTn} Under Assumptions 1--8, and if $a_{n,\lambda_{2}}\to 0$ and $b_{n,\lambda_{2}}\to 0$ as $n\to\infty$, then there exists a local solution $\widehat{\bs{\beta}}$ in (\ref{DEF:penalized-min}) such that $\|\widehat{\bs{\beta}}-\bs{\beta}_{0}\|=O_{P}(n^{-1/2}+a_{n,\lambda_{2}})$.
\end{theorem}

Next we define $\mathbf{\kappa }_{n,\lambda_{2}}=\{p_{\lambda _{2}}^{\prime }(|\beta _{10}|)\mathrm{sgn}(\beta _{10}),\cdots ,p_{\lambda _{2}}^{\prime }(|\beta _{q0}|)\mathrm{sgn}
(\beta _{q0})\}^{\top}$ and a diagonal matrix $\mathbf{\Sigma}_{\lambda_{2}}=\mathrm{diag}
\{p_{\lambda _{2}}^{\prime \prime }(|\beta _{10}|),\cdots ,p_{\lambda_{2}}^{\prime \prime }(|\beta _{q0}|)\}$. The theorem below shows that under regularity conditions, all the covariates with zero coefficients can be detected simultaneously with probability tending to one, and the estimators of all the nonzero coefficients are asymptotically normally distributed. 
\begin{theorem}
	\label{THM:ORACLE} Under Assumptions 1--8, if
	$\lim_{n\rightarrow \infty}\sqrt{n}\lambda _{2}\rightarrow \infty$,  and\\ $
	\liminf_{n\rightarrow \infty }\liminf_{\beta _{k}\rightarrow 0^{+}}\lambda
	_{2}^{-1}p_{\lambda _{2}}^{\prime }(|\beta _{k}|)>0$, then the $\sqrt{n}$-consistent estimator $\widehat{\bs{\beta}}$
	in Theorem \ref{THM:ROOTn} satisfies $P(\widehat{\bs{\beta}}_{2}=\mathbf{0})\rightarrow 1$, as $n\rightarrow \infty $, and
	\begin{equation*}
	\sqrt{n}(\bs{\Sigma}_{s} +\bs{\Sigma} _{\lambda_{2} }) \left\{ \widehat{\bs{\beta}}_{1}-\bs{\beta}_{10}+(\bs{\Sigma}_{s} +\bs{\Sigma} _{\lambda_{2} }) ^{-1}\bs{\kappa}_{n,\lambda_{2}}\right\} \rightarrow \mathrm{N}(\mathbf{0},\sigma^{2}\bs{\Sigma}_{s}),
	\end{equation*}
	where
	\begin{equation}
	\bs{\Sigma}_{s}=\sigma^{-2}E[(\mathbf{Z}_{1}-\widetilde{\mathbf{Z}}_{1})
	(\mathbf{Z}_{1}-\widetilde{\mathbf{Z}}_{1})^{\top}]
	\label{DEF:Sigma}
	\end{equation}
	with $\widetilde{\mathbf{Z}}_{1}$ given in (\ref{DEF:Z1tilde}).
	
\end{theorem}
The next result provides the global convergence of the nonparametric estimator $\widehat{\alpha}(\cdot)$.
\begin{corollary}
	\label{THM:g-convergence}
	Suppose Assumptions 1--8 hold, then the bivariate penalized estimator $\widehat{\alpha}(\cdot)$, given in (\ref{DEF:g_hat}), is consistent with the true function, $\alpha_{0}$, and satisfies that
	\[
	\Vert \widehat{\alpha}-\alpha_{0}\Vert _{L^2}=O_{P}\left\{
	\frac{\lambda_{1} }{n\left| \triangle \right| ^{3}}|\alpha_{0}|_{2,\infty}
	+\left(1+\frac{\lambda_{1} }{n\left| \triangle \right| ^{5}}\right)|\triangle |^{\ell +1}|\alpha_{0}|_{\ell,\infty}+\frac{1}{\sqrt{n}|\triangle|}
	\right\}.
	\]
\end{corollary}
This is a direct result from \cite{Wang:Wang:Lai:Gao:18}, thus the proof is omitted.

\setcounter{chapter}{4}\setcounter{equation}{0} 
\renewcommand{\thetable}{4.\arabic{table}} \setcounter{table}{0} 
\renewcommand{\thefigure}{4.\arabic{figure}} \setcounter{figure}{0} 
\vskip .12in \noindent \textbf{4. Implementation} \label{SEC:implementation} \vskip 0.10in

Since the SCAD penalty function is singular at the origin, and it does not have continuous second order derivatives. To solve the minimization problem in (\ref{DEF:penalized-min}), one can locally approximate it by a quadratic function \citep{Fan:Li:01,Lian:12}, then the minimization problem of $R(\bs{\beta}; \lambda_{1},\lambda_{2})$ can be solved using quadratic minimization. However, employing the local quadratic approximation can be extremely expensive since it requires the repeated factorization of large matrices repeatedly for different smoothing parameters. In addition, quadratic minimization is not able to provide naturally sparse estimates. In the implementation of our method, we consider the use of the coordinate descent algorithm \citep{Breheny:Huang:15}, which fits the penalized regressions more stably and efficiently.

The classical coordinate descent algorithm deals with the optimization problem with one tuning parameter, and there are several ways to address the double-penalization. A natural idea is to solve the optimization problem by searching over a 2D grid for tuning parameters, which can be computationally expensive. We propose the following algorithm based on coordinate descent:
\begin{itemize}
	\item[Step 0.] Obtain $\widetilde{\bs{\pi}}$ by minimizing objective function w.r.t. $\bs{\pi}$:
	$
	\frac{1}{2}\|\mathbf{Y}-\mathbf{B}\mathbf{Q}_{2}\bs{\pi}\|^{2}
	+\lambda_{0}(\mathbf{Q}_{2}\bs{\pi})^{\top}\mathbf{P}(\mathbf{Q}_{2}\bs{\pi})
	$
	with $\lambda_0$ selected via GCV, and obtain $\widetilde{\mathbf{Y}}=\mathbf{B}\mathbf{Q}_{2}\widetilde{\bs{\pi}}$ and $\widetilde{\mathbf{Z}}=\mathbf{H}_\mathbf{B}(\lambda_0)\mathbf{Z}$;\
	
	\item[Step 1.] Obtain $\widehat{\bs{\beta}}$ by minimizing objective function w.r.t. $\bs{\beta}$:
	$
	\frac{1}{2}\|\mathbf{Y}-\widetilde{\mathbf{Y}}-(\mathbf{Z}-\widetilde{\mathbf{Z}})\bs{\beta}\|^{2}
	+n\sum_{j=1}^{p}p_{\lambda_{2}}(|\beta_{j}|)
	$
	with $\lambda_2$ selected via BIC;
	
	\item[Step 2.] Let $\mathbf{Z}^{*}$ be the selected covariates from Step 1. Based on data $\{(\mathbf{Z}_{i}^{*},\mathbf{X}_{i},Y_i)\}_{i=1}^{n}$ refit model (\ref{model})  to obtain $\widehat{\bs{\beta}}$ and $\widehat{\bs{\theta}}$ by minimizing the following objective function w.r.t. $\bs{\beta}$ and $\bs{\theta}$:
	$
	\|\mathbf{Y}-\mathbf{Z}^{*}\bs{\beta}-\mathbf{B}\mathbf{Q}_{2}\bs{\theta}\|^{2}
	+\lambda_{1}(\mathbf{Q}_{2}\bs{\theta})^{\top}\mathbf{P}(\mathbf{Q}_{2}\bs{\theta}).
	$
\end{itemize}

\vskip .12in \noindent \textbf{4.1. Standard error formula} \vskip .10in

The standard errors for the estimated parameters can be obtained directly because we are estimating parameters and selecting variables at the same time. Note that for any $\lambda_{1}$ and $\lambda_{2}$ the fitted values at the $n$ data points are
$
\widehat{\mathbf{Y}}=\mathbf{Z}\widehat{\bs{\beta}}+\mathbf{B} \mathbf{Q}_{2} \bs{\theta}(\widehat{\bs{\beta}})
=\mathbf{S}(\lambda_{1},\lambda_{2})\mathbf{Y}
$, 
where $\bs{\theta}(\bs{\beta})$ is given in (\ref{EQ:theta(beta)}).  Therefore, the smoothing or hat matrix can be written as
\begin{align}
\mathbf{S}&(\lambda_{1},\lambda_{2})=\left(\begin{array}{cc}
\mathbf{Z}-\widehat{\mathbf{Z}} \!&\!\mathbf{B}\mathbf{Q}_{2}
\end{array}\right) \notag \\
& \times \left( \begin{array}{cc}
\{(\mathbf{Z}-\widehat{\mathbf{Z}})^{\top}(\mathbf{Z}-\widehat{\mathbf{Z}}) + n\bs{\Sigma}_{\lambda_{2}}(\widehat{\bs{\beta}})\}^{-1}&  \mathbf{0}\\
\mathbf{0}  & \{\mathbf{Q}_{2}^{\top}(\mathbf{B}^{\top}\mathbf{B}+\lambda_{1}\mathbf{P})\mathbf{Q}_{2}\}^{-1}
\end{array} \right) \notag
\left(\begin{array}{c}
\mathbf{Z}^{\top}-\widehat{\mathbf{Z}}^{\top}\\
\mathbf{Q}_{2}^{\top}\mathbf{B}^{\top}
\end{array}\right), \notag
\end{align}
where $\widehat{\mathbf{Z}}=\mathbf{H}_{\mathbf{B}}(\lambda_{1})\mathbf{Z}$ and $\bs{\Sigma}_{\lambda_{2}}(\bs{\beta})\approx \mathrm{diag}\left\{{p_{\lambda_{2}}'(|\beta_{1}|)}/{|\beta_{1}|},
\ldots,{p_{\lambda_{2}}'(|\beta_{p}|)}/{|\beta_{p}|}\right\}$.

Finally, we derive a sandwich formula for the standard error of $\widehat{ \bs{\beta}}$
\begin{align*}
\widehat{\mathrm{Cov}}(\widehat{\bs{\beta}})
=&\widehat{\sigma}^{2}\left\{(\mathbf{Z}-\widehat{\mathbf{Z}})^{\top}
(\mathbf{Z}-\widehat{\mathbf{Z}}) +n \bs{\Sigma} _{\lambda_{2}}( \widehat{\bs{
		\beta}}) \right\} ^{-1}(\mathbf{Z}-\widehat{\mathbf{Z}})^{\top} (\mathbf{Z}-\widehat{\mathbf{Z}}) \\
&\times \left\{ (\mathbf{Z}-\widehat{\mathbf{Z}})^{\top}(\mathbf{Z}-\widehat{\mathbf{Z}})
+n\bs{\Sigma} _{\lambda_{2}}( \widehat{\bs{\beta }}) \right\}^{-1},  \notag
\label{EQ:std}
\end{align*}
where $\widehat{\sigma}^{2}=\|\mathbf{Y}-\widehat{\mathbf{Y}}\|^{2}/
\{n-\mathrm{tr}(\mathbf{S}(\lambda_{1},\lambda_{2}))\}$.
Applying conventional techniques that arise in the bivariate splines setting, we can show that the above sandwich formula is a consistent estimator and has good accuracy in our simulation study for moderate sample sizes.

\setcounter{chapter}{5} \renewcommand{\thetheorem}{5.\arabic{theorem}}
\renewcommand{\thelemma}{5.\arabic{lemma}}
\renewcommand{\theproposition}{5.\arabic{proposition}}
\renewcommand{\thetable}{5.\arabic{table}} \setcounter{table}{0} 
\renewcommand{\thefigure}{5.\arabic{figure}} \setcounter{figure}{0} 
\setcounter{equation}{0} \setcounter{lemma}{0} \setcounter{theorem}{0}
\setcounter{proposition}{0}\setcounter{corollary}{0}
\vskip .12in \noindent \textbf{5. Simulation} \vskip 0.10in
\label{sec:application} \label{SEC:simulation}

In this section, we conduct Monte Carlo simulation studies to evaluate the finite-sample performance of the proposed doubly-penalized method in terms of both model estimation and variable selection. We compare our method (PLSM) with the spatial weighted regression method (SWR) proposed by \cite{Nandy:Lim:Maiti:17} and linear model method (LM).

\vskip .12in \noindent \textbf{5.1. Example 1} \vskip .10in
\label{SUBSEC:example1}

In this example, we consider a modified horseshoe shaped domain $\Omega$ with the surface test function used by \cite{Wood:Bravington:Hedley:08}. First, we generated 80$\times$180 grid points over the domain. Then, for 100 Monte Carlo experiments, we randomly sample $n$ grid points on $\Omega$ with $n=100$ or $200$. The response variable $Y_{i}$'s are generated from the following PLSM:
$Y_{i}=\mathbf{Z}_{i}^{\top}\bs{\beta}+\alpha(\mathbf{X}_{i})+\varepsilon_{i}$, $i=1,\ldots,n$,
where the true coefficients are $\bs{\beta}=(1,-1,0,0,0,0,0,0)^{\top}$ and $\varepsilon_{i},~i=1,\ldots,n$ are generated independently from $N(0,\sigma^2)$ with $\sigma=0.2$. Figure \ref{FIG:eg1-1} (a) and (b) show the surface plot and the contour map of the true function $\alpha(\cdot)$, respectively. Note that the design of the function $\alpha(\cdot)$ makes it hard to have a linear approximation or nonlinear additive approximation of $\alpha(\cdot)$ on a rectangular domain. As a result, many traditional parametric and nonparametric methods do not work well in this case.

\begin{figure}[htbp]
	\begin{center}
		\begin{tabular}{cc}
			\includegraphics[height=4.5cm]{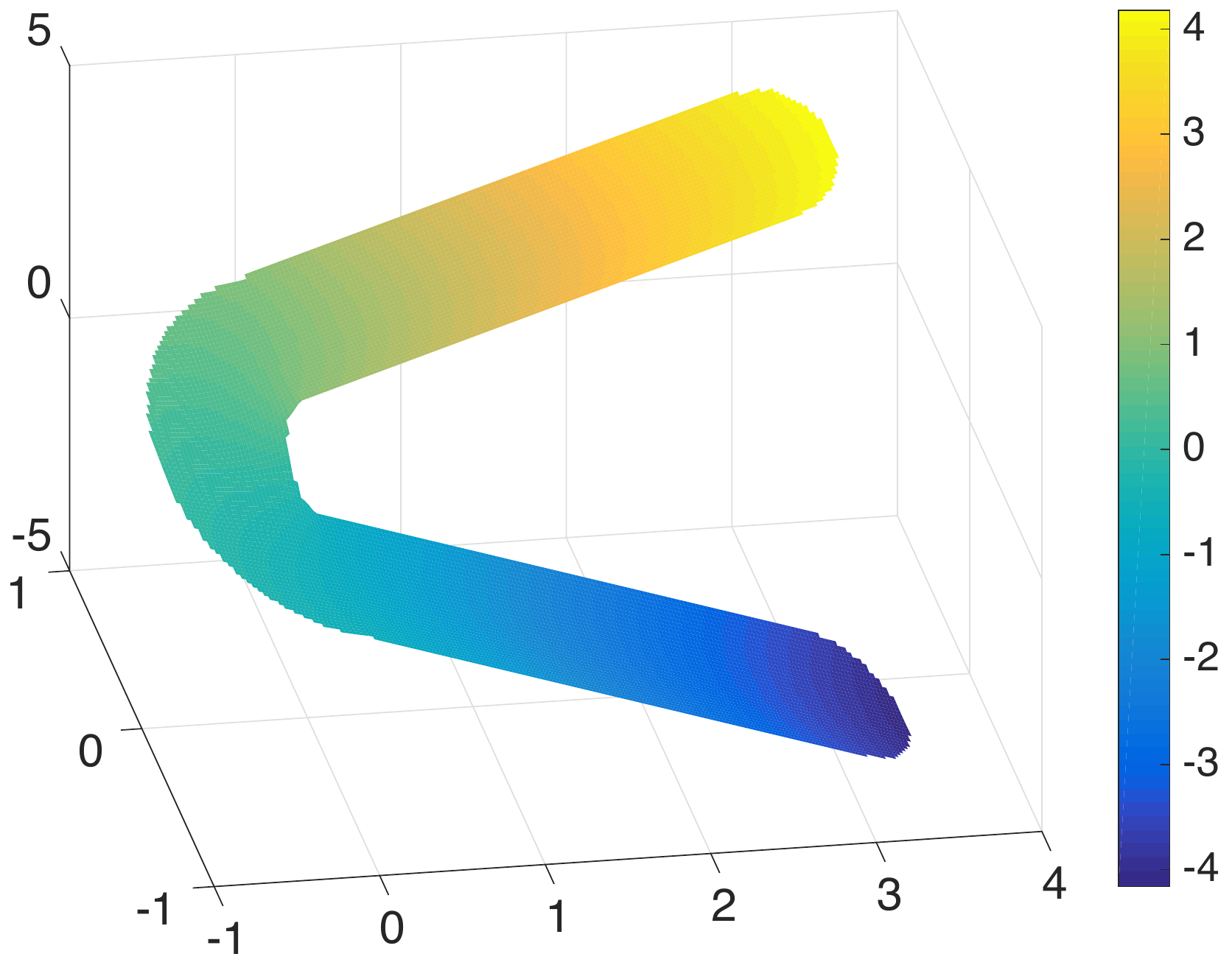} &\includegraphics[height=4.3cm]{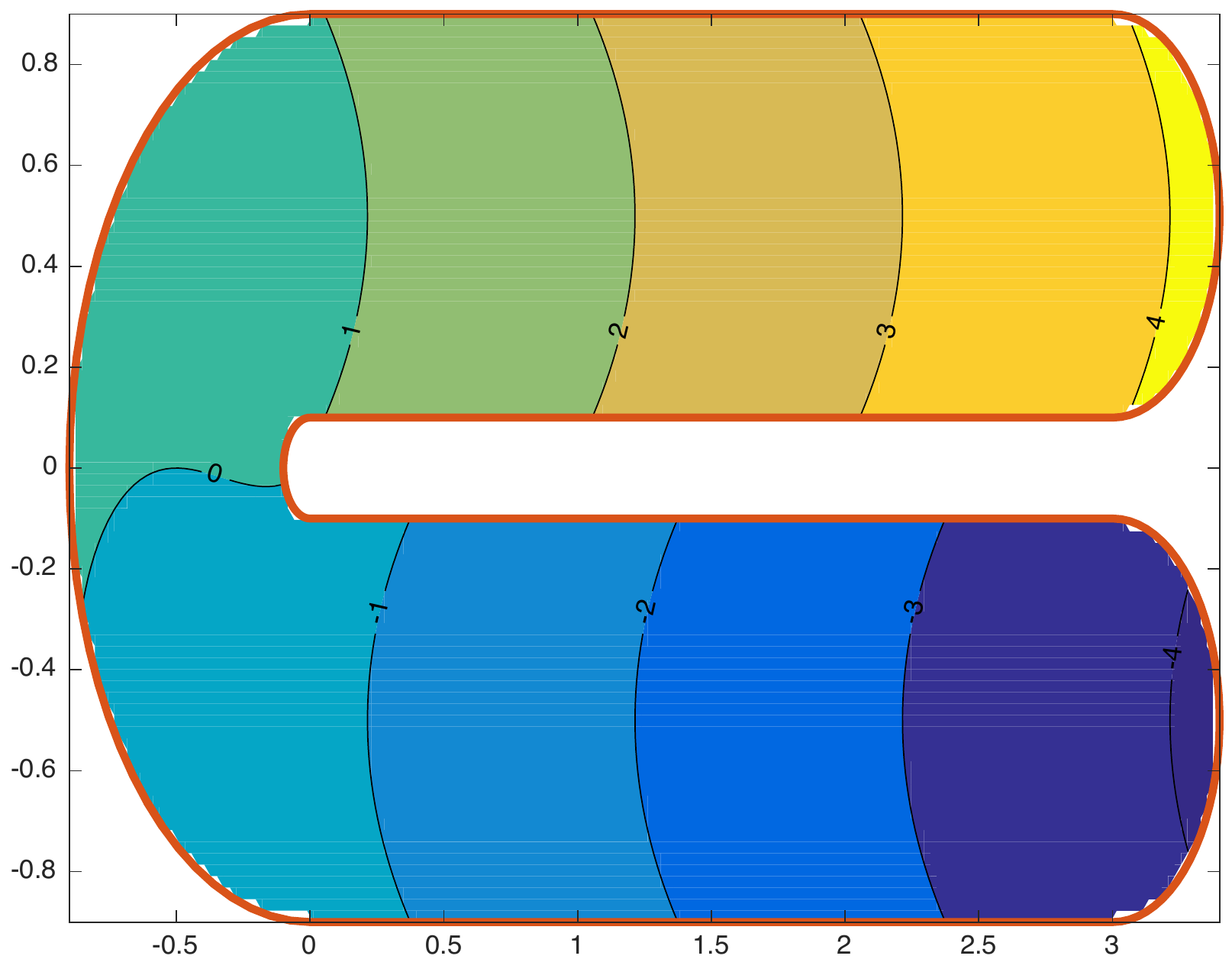}\\
			(a) &(b)\\
		\end{tabular}
	\end{center}
	\caption{Example 1. (a) true function of $\alpha(\cdot)$; (b) contour map of true function $\alpha(\cdot)$.}
	\label{FIG:eg1-1}
\end{figure}

In practice, some covariates may vary over space, that is, they may be correlated with spatial locations. To study the performance of variable selection at different correlation levels, similar as in \cite{Wang:Wang:Lai:Gao:18}, we generate the covariates as follows:
$Z_{i1}=-\frac{2}{3}\left\{\arctan{\pi\left(\rho\frac{X_{i1}}{X_{i2}}+(1-\rho)U_{i}\right)}\right\}$,
$Z_{i3}=\cos{\pi\left(\rho\frac{X_{i1}}{X_{i2}}+(1-\rho)U_{i}\right)}$, 
$Z_{ij} \sim\mathrm{Uniform}(-1,1)$, $j=2,4,\ldots,8$, $U_{i}\sim \mathrm{Uniform}(-1,1)$.
In particular, we consider the following three cases: (i) low correlation ($\rho=0.3$); (ii) medium correlation ($\rho=0.5$); and (iii) high correlation ($\rho=0.7$).

Figure \ref{FIG:eg1-2} (a) demonstrates the sampled location points of replicate 1. For the bivariate spline approximation, we consider three different triangulations on the horseshoe domain with (i) 90 triangles and 74 vertices; (ii) 158 triangles and 114 vertices; and (iii) 286 triangles and 186 vertices as illustrated in Figure \ref{FIG:eg1-1} (b)--(d), respectively.

Columns 4-6 in Table \ref{TAB:eg1-1} report the average number of two nonzero coefficients incorrectly set to zero (denoted as ``F''), the average number of six zero coefficients correctly set to zero (denoted as ``T''), and how often a correct model is chosen among 100 replications (denoted as ``C'').  We compare the sparse PLSM ($\mathcal{S}$-PLSM) estimator with the ``oracle" estimator (ORACLE), the estimator when the true model is known prior to statistical analysis. In this example, the ORACLE is calculated using triangulation $\triangle_2$.  We also compare the $\mathcal{S}$-PLSM with the sparse spatially weighted regression method ($\mathcal{S}$-SWR) proposed by \cite{Nandy:Lim:Maiti:17}. From Table \ref{TAB:eg1-1}, one sees that, the proposed method performs very well regardless of the level of correlation, and the ``F'', ``T'' and ``C'' are very close to the ORACLE. However, the $\mathcal{S}$-SWR is very sensitive to the correlation level between the covariates and spatial locations. When some of the covariates are highly correlated with the spatial locations, the correct selection rate of the $\mathcal{S}$-SWR is low, especially when the sample size is small. The $\mathcal{S}$-PLSM selection results also indicate that the number of triangles has little effect on the performance of variable selection.

\begin{figure}[htbp]
	\begin{center}
		\begin{tabular}{cc}
			\includegraphics[height=3.4cm]{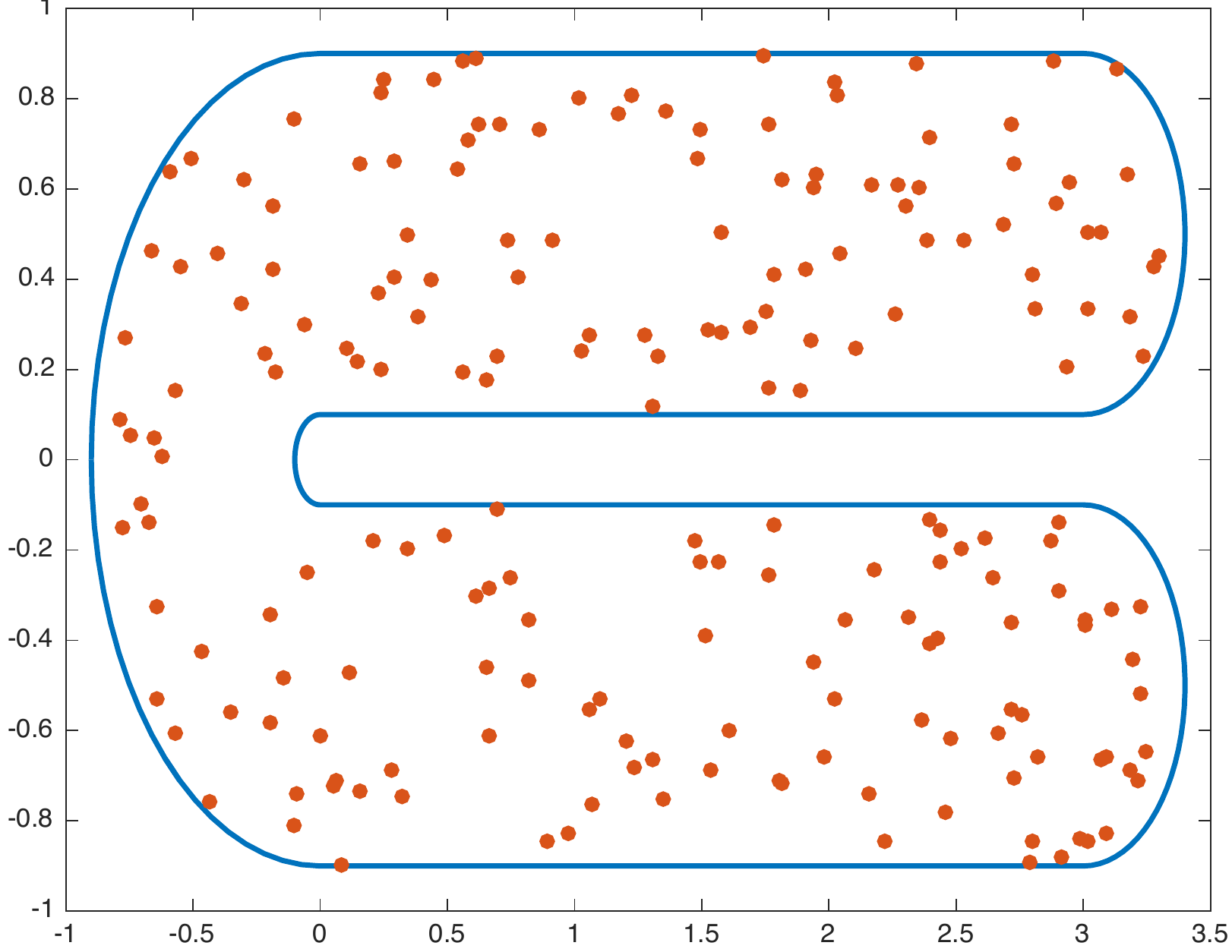} &\includegraphics[height=3.4cm]{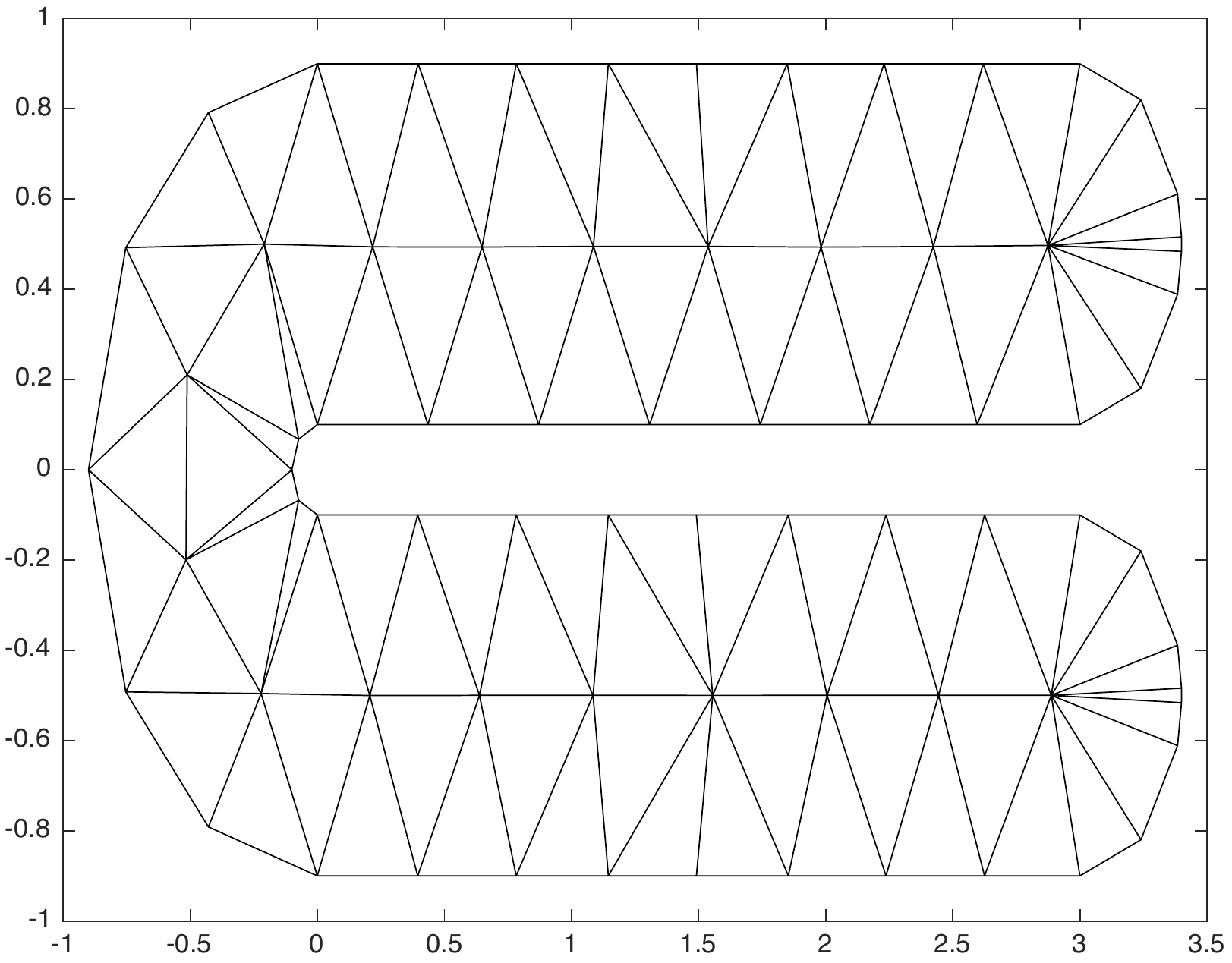}\\[-6pt]
			(a) &(b)\\
			\includegraphics[height=3.4cm]{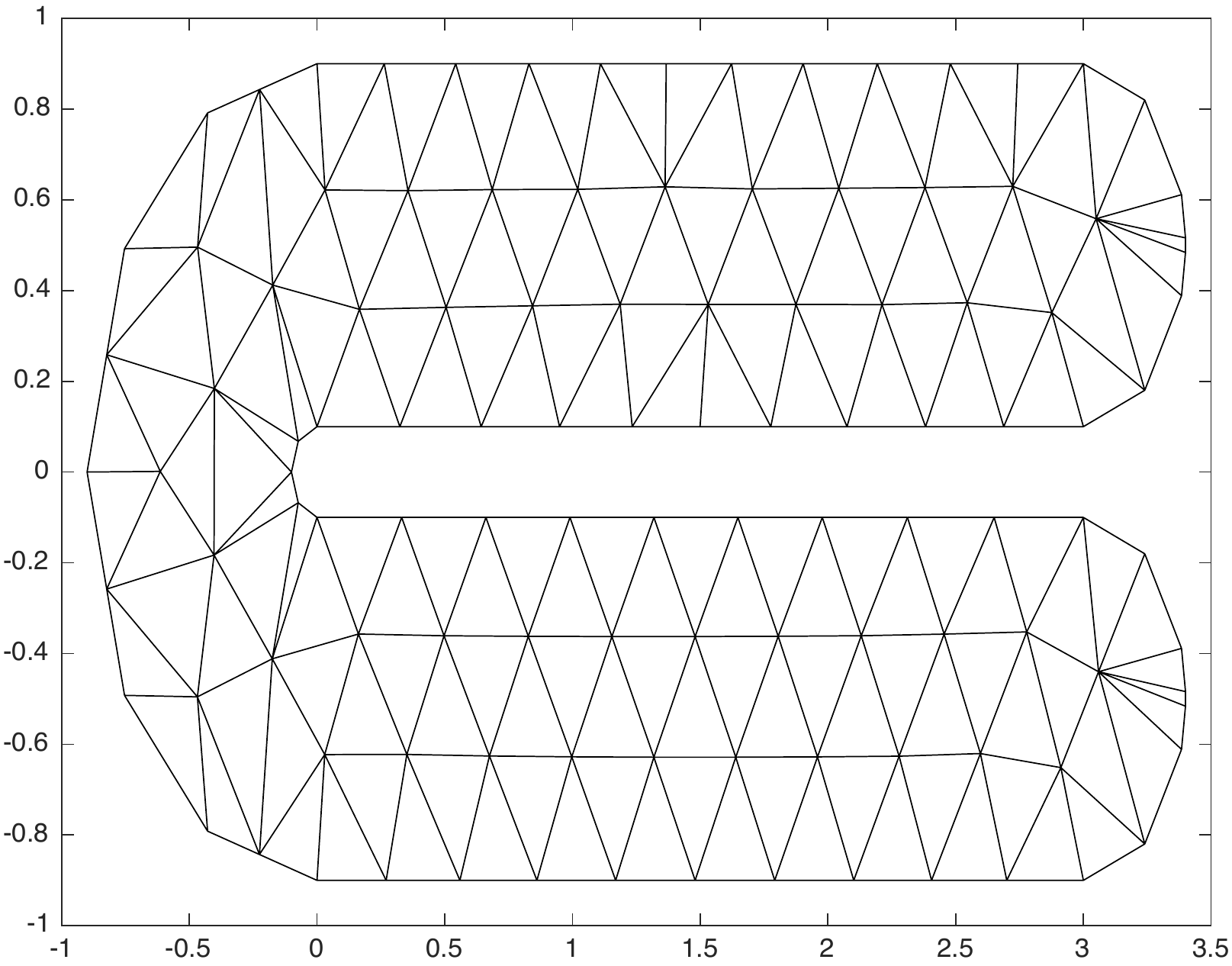} &\includegraphics[height=3.4cm]{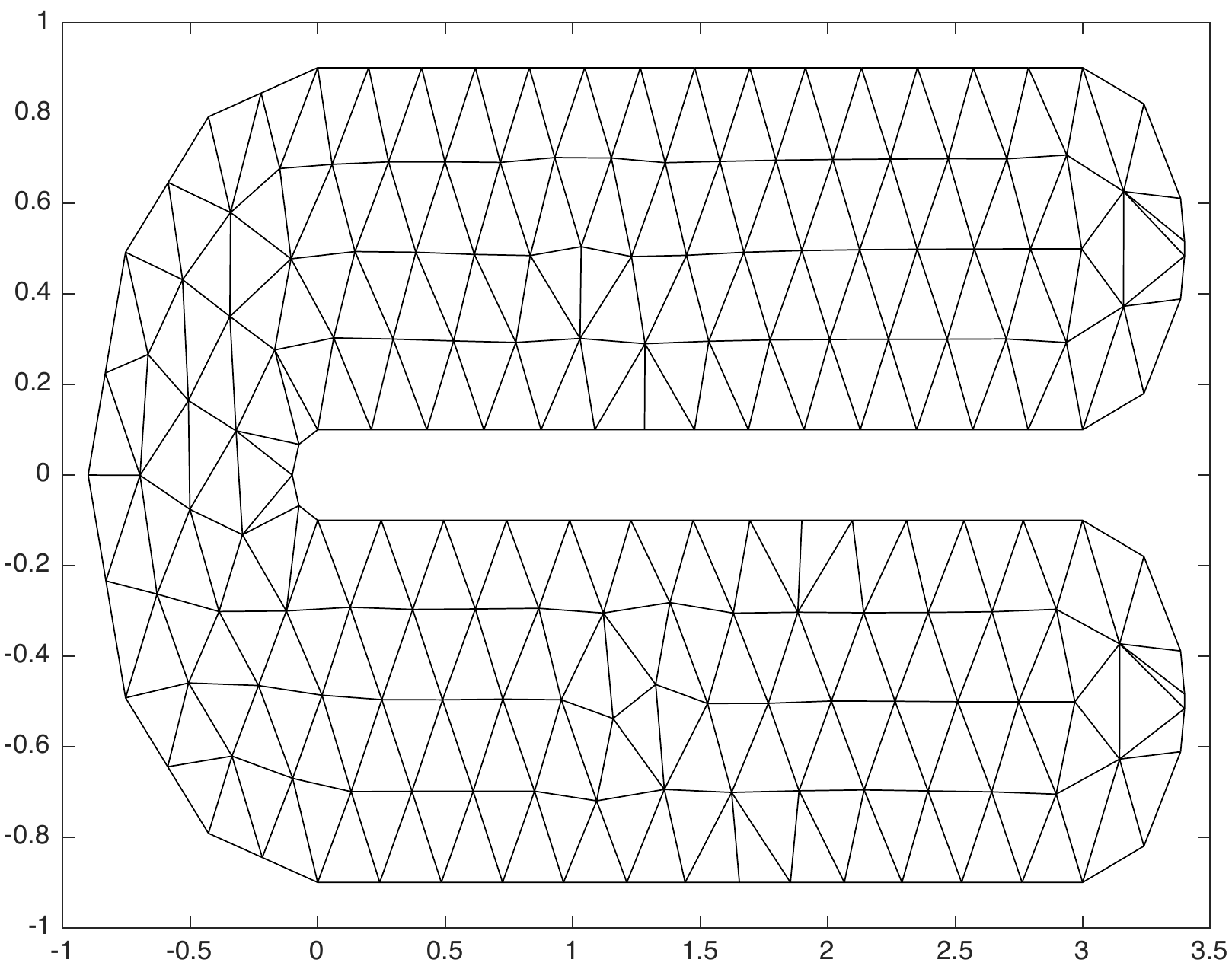}\\[-6pt]
			(c) &(d)\\
		\end{tabular}
	\end{center}
	\caption{Example 1. (a) sampled location points of replicate 1; (b) $\triangle_1$ over the domain; (c) $\triangle_2$ over the domain and (d) $\triangle_3$ over the domain.}
	\label{FIG:eg1-2}
\end{figure}

Next, to see the accuracy of the estimators, we compute the root mean squared error (RMSE) for each of the estimators based on 100 Monte Carlo samples and compare them with the ORACLE estimator. Columns 7-9 in Table \ref{TAB:eg1-1} show the RMSEs of the estimate of the parameters $\beta_{1}$, $\beta_{2}$ as well as the nonlinear function $\alpha(\cdot)$. In general, the table clearly indicates that the proposed method estimates unknown parameters and function very well even when the correlation is high. Regardless of the choice of triangulation, the $\mathcal{S}$-PLSM with the SCAD penalty always provides accurate estimators in the sense that they are very close to the ``ORACLE". Figure \ref{FIG:eg1-3} shows the estimator of $\alpha(\cdot)$ using different triangulations with the SCAD penalty for a typical data with $n=200$ observations generated from different correlation levels. The proposed PLSM estimator looks globally close to the true surface regardless of the $\rho$ used.

\begin{table}
	\begin{center}
		\caption{Example 1. model selection and estimation results.}
		\label{TAB:eg1-1}
		\scalebox{0.8}{\begin{tabular*}{\columnwidth}{@{\extracolsep{\fill}}ccccccccc}\hline\hline
				\multirow{2}{*}{$\rho$} & \multirow{2}{*}{$n$} &\multirow{2}{*}{Method} &\multicolumn{3}{c}{Selection}
				&\multicolumn{3}{c}{RMSE} \\ \cline{4-9}
				& & &F &T &C &$\beta_{1}$ &$\beta_{2}$ &$\alpha(\cdot)$\\ \hline
				
				\multirow{10}{*}{0.3} &\multirow{5}{*}{100} &ORACLE &0.00 &6.00 &100 &0.103 &0.041 &0.137\\
				& &$\mathcal{S}$-SWR &0.39 &5.81 &48 &0.823 &0.416 &--\\
				& &$\mathcal{S}$-PLSM-$\triangle_1$ &0.00 &5.86 &87 &0.082 &0.049 &0.125\\
				& &$\mathcal{S}$-PLSM-$\triangle_2$ &0.00 &5.94 &95 &0.107 &0.041 &0.138\\
				& &$\mathcal{S}$-PLSM-$\triangle_3$ &0.00 &5.86 &89 &0.085 &0.049 &0.126\\ \cline{2-9}
				
				&\multirow{5}{*}{200} &ORACLE &0.00 &6.00 &100 &0.066 &0.027 &0.104\\
				& &$\mathcal{S}$-SWR &0.00 &5.95 &96 &0.507 &0.419 &--\\
				& &$\mathcal{S}$-PLSM-$\triangle_1$ &0.00 &5.90 &95 &0.052 &0.032 &0.097\\
				& &$\mathcal{S}$-PLSM-$\triangle_2$ &0.00 &5.98 &98 &0.066 &0.027 &0.104\\
				& &$\mathcal{S}$-PLSM-$\triangle_3$ &0.00 &5.90 &95 &0.052 &0.032 &0.096\\ \hline
				
				\multirow{10}{*}{0.5} &\multirow{5}{*}{100} &ORACLE &0.00 &6.00 &100 &0.095 &0.041 &0.132\\
				& &$\mathcal{S}$-SWR &0.87 &5.91 &9 &0.999 &0.420 &-- \\
				& &$\mathcal{S}$-PLSM-$\triangle_1$ &0.00 &5.89 &90 &0.099 &0.042 &0.136\\
				& &$\mathcal{S}$-PLSM-$\triangle_2$ &0.00 &5.87 &90 &0.095 &0.041 &0.132 \\
				& &$\mathcal{S}$-PLSM-$\triangle_3$ &0.00 &5.82 &86 &0.117 &0.042 &0.148\\ \cline{2-9}
				
				&\multirow{5}{*}{200} &ORACLE &0.00 &6.00 &100 &0.066 &0.028 &0.104\\
				& &$\mathcal{S}$-SWR &0.32 &5.80 &50 &0.814 &0.424 &--\\
				& &$\mathcal{S}$-PLSM-$\triangle_1$ &0.00 &5.98 &98 &0.055 &0.032 &0.099\\
				& &$\mathcal{S}$-PLSM-$\triangle_2$ &0.00 &5.95 &97 &0.066 &0.028 &0.104\\
				& &$\mathcal{S}$-PLSM-$\triangle_3$ &0.00 &5.92 &96 &0.055 &0.032 &0.098\\ \hline
				
				\multirow{10}{*}{0.7} &\multirow{5}{*}{100} &ORACLE &0.00 &6.00 &100 &0.132 &0.041 &0.161\\
				& &$\mathcal{S}$-SWR &0.90 &5.92 &8 &1.001 &0.420 &--\\
				& &$\mathcal{S}$-PLSM-$\triangle_1$ &0.00 &5.86 &89 &0.141 &0.048 &0.164\\
				& &$\mathcal{S}$-PLSM-$\triangle_2$ &0.00 &5.83 &89 &0.159 &0.041 &0.179\\
				& &$\mathcal{S}$-PLSM-$\triangle_3$ &0.00 &5.89 &92 &0.154 &0.049 &0.173\\ \cline{2-9}
				
				&\multirow{5}{*}{200} &ORACLE &0.00 &6.00 &100 &0.076 &0.027 &0.110\\
				& &$\mathcal{S}$-SWR &0.70 &5.93 &25 &1.129 &0.418 &--\\
				& &$\mathcal{S}$-PLSM-$\triangle_1$ &0.00 &5.95 &96 &0.077 &0.031 &0.110\\
				& &$\mathcal{S}$-PLSM-$\triangle_2$ &0.00 &5.99 &99 &0.076 &0.027 &0.110\\
				& &$\mathcal{S}$-PLSM-$\triangle_3$ &0.00 &5.94 &95 &0.075 &0.031 &0.108\\ \hline\hline
		\end{tabular*}}
	\end{center}
\end{table}

\begin{figure}[htbp]
	\begin{center}
		\begin{tabular}{cccc}
			&$\triangle_1$ &$\triangle_2$ &$\triangle_3$\\
			$\rho=0.3$ & \includegraphics[height=3cm]{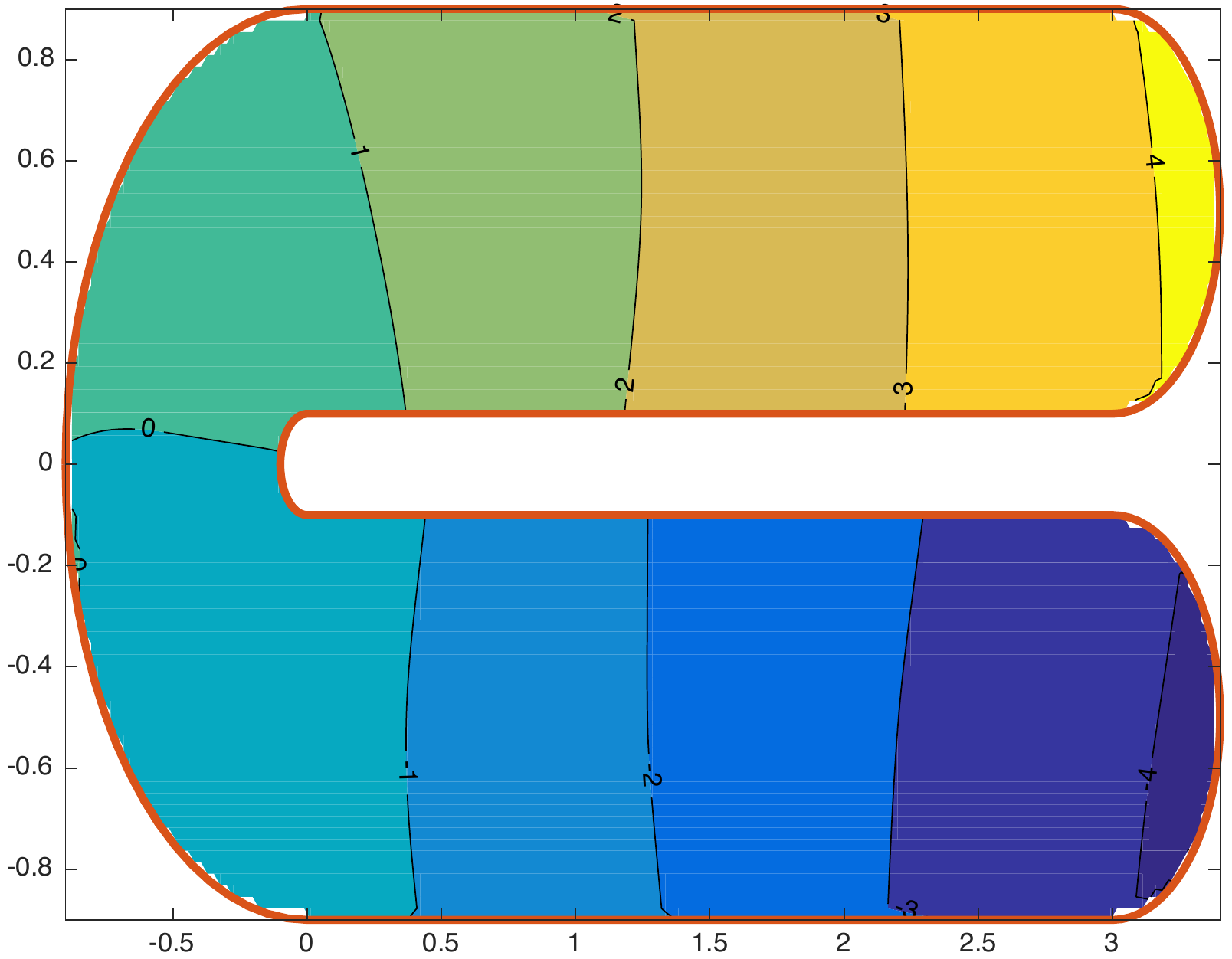} &\includegraphics[height=3cm]{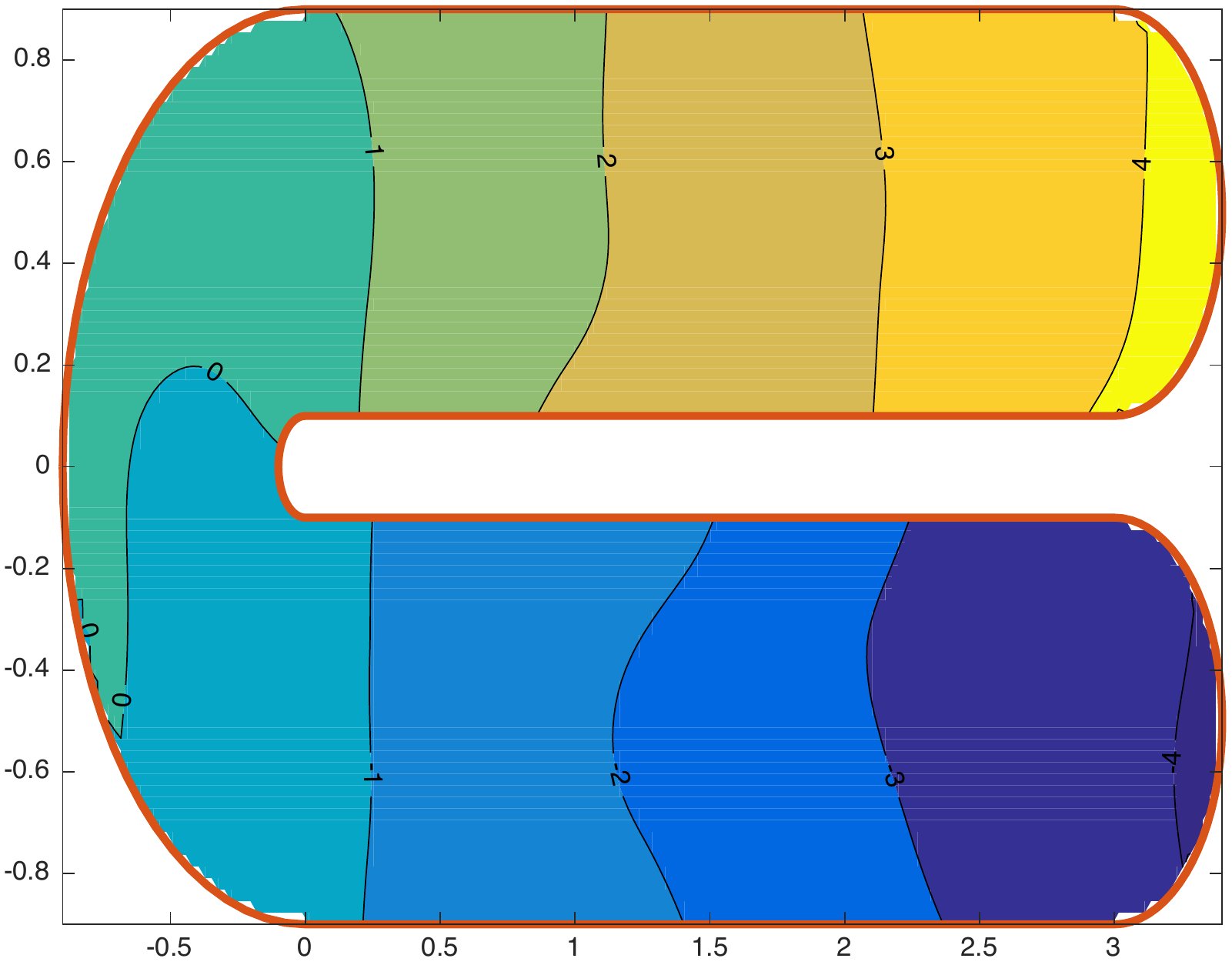} &\includegraphics[height=3cm]{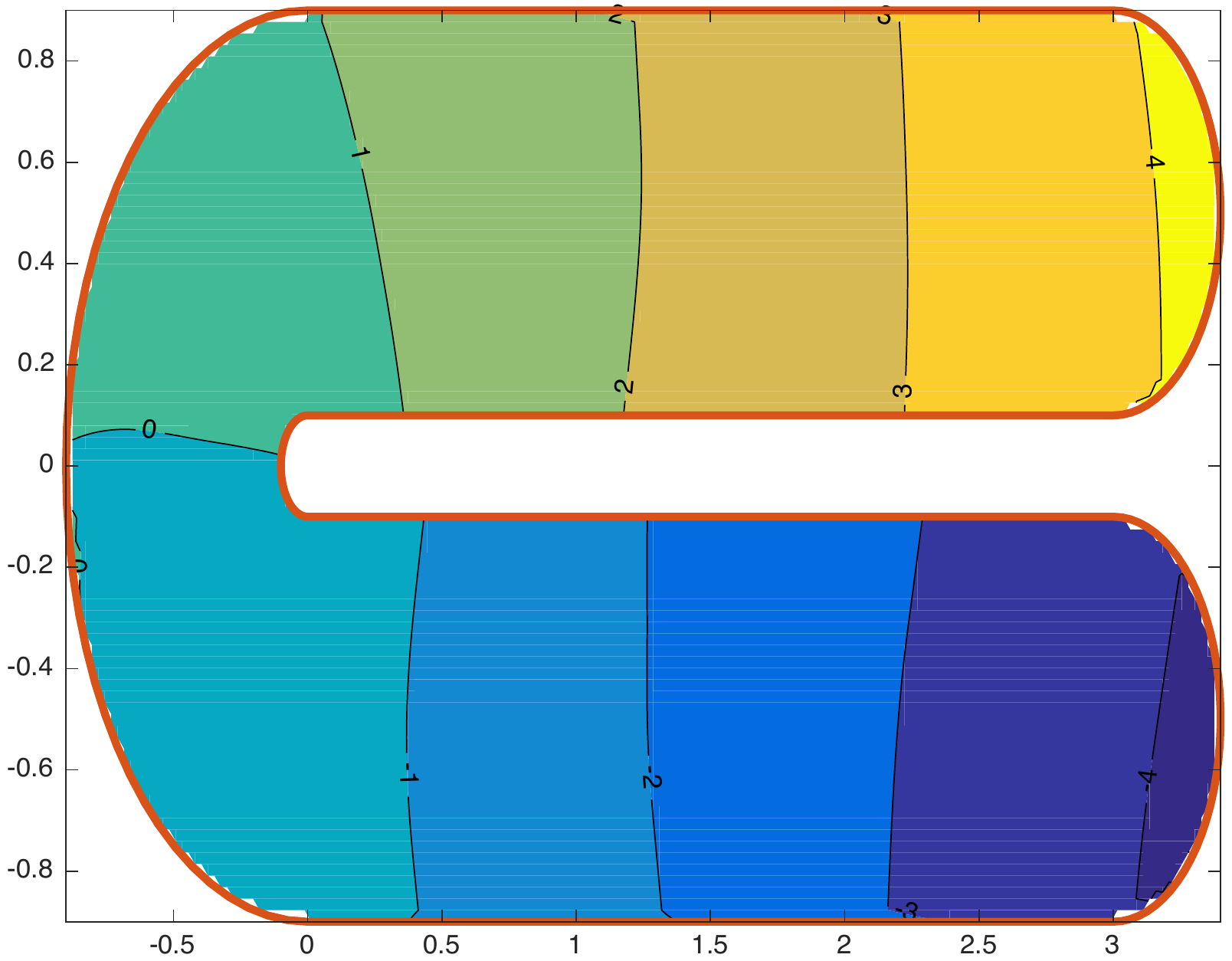}\\ [-6pt]
			$\rho=0.5$ & \includegraphics[height=3cm]{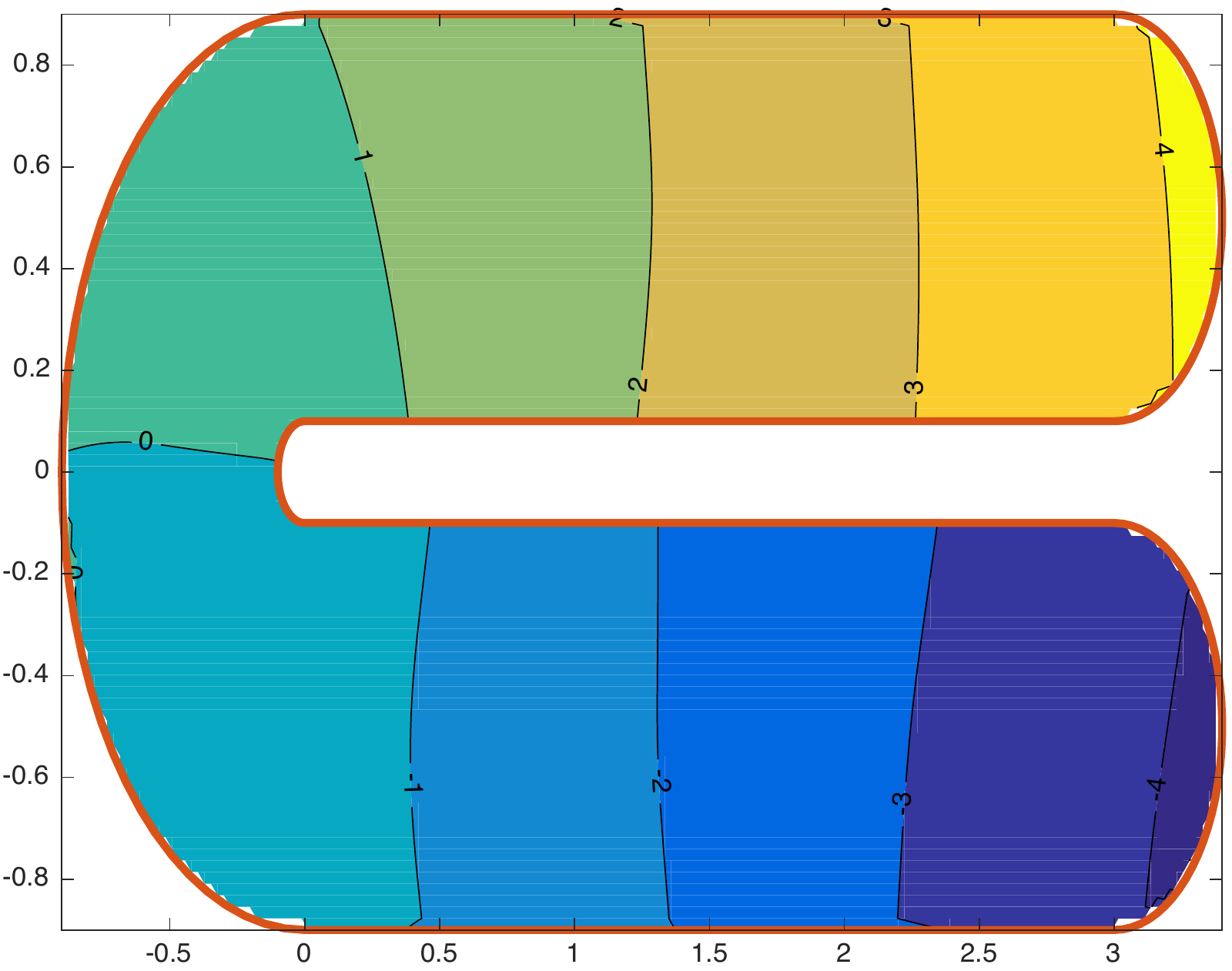} &\includegraphics[height=3cm]{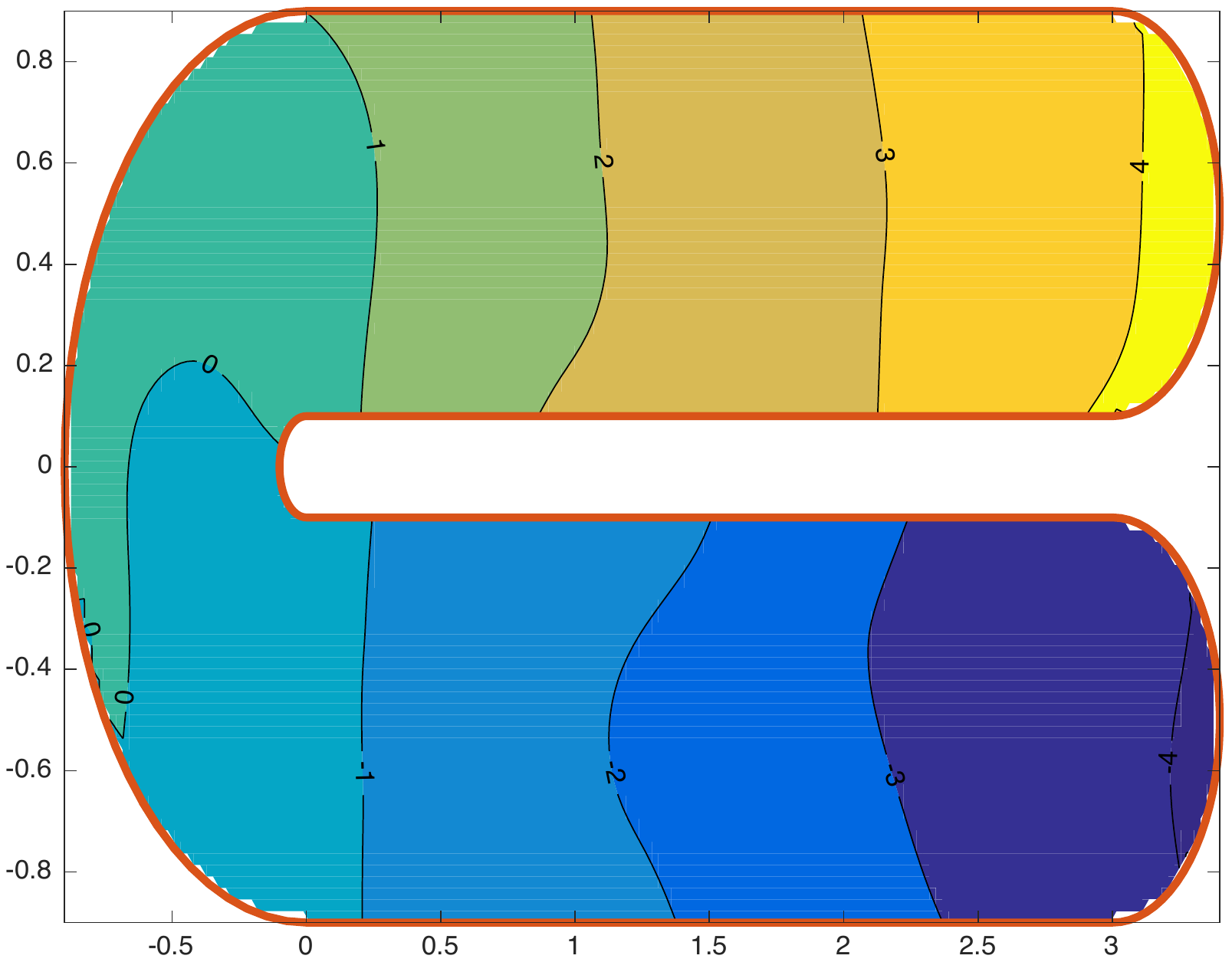} &\includegraphics[height=3cm]{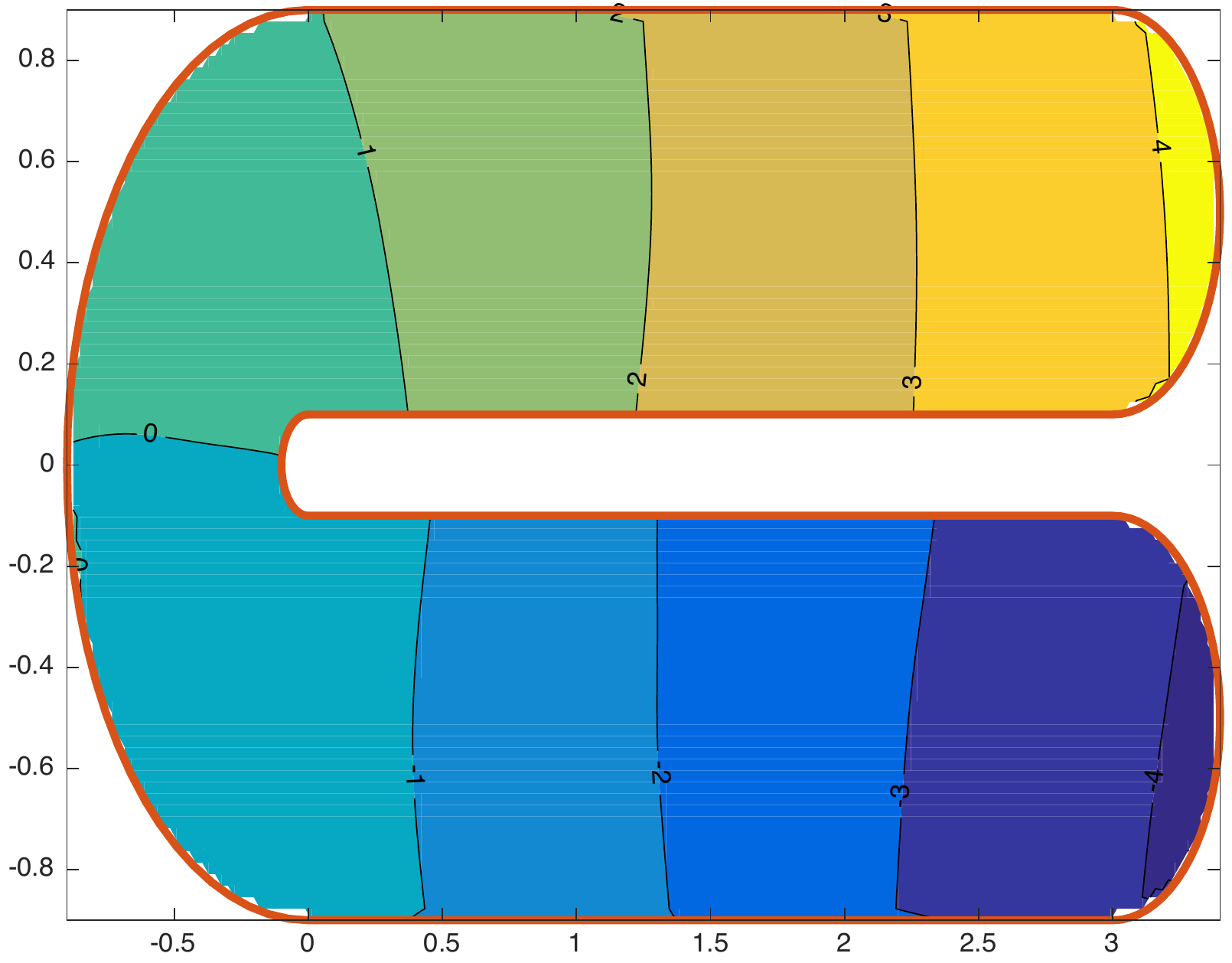}\\ [-6pt]
			$\rho=0.7$ & \includegraphics[height=3cm]{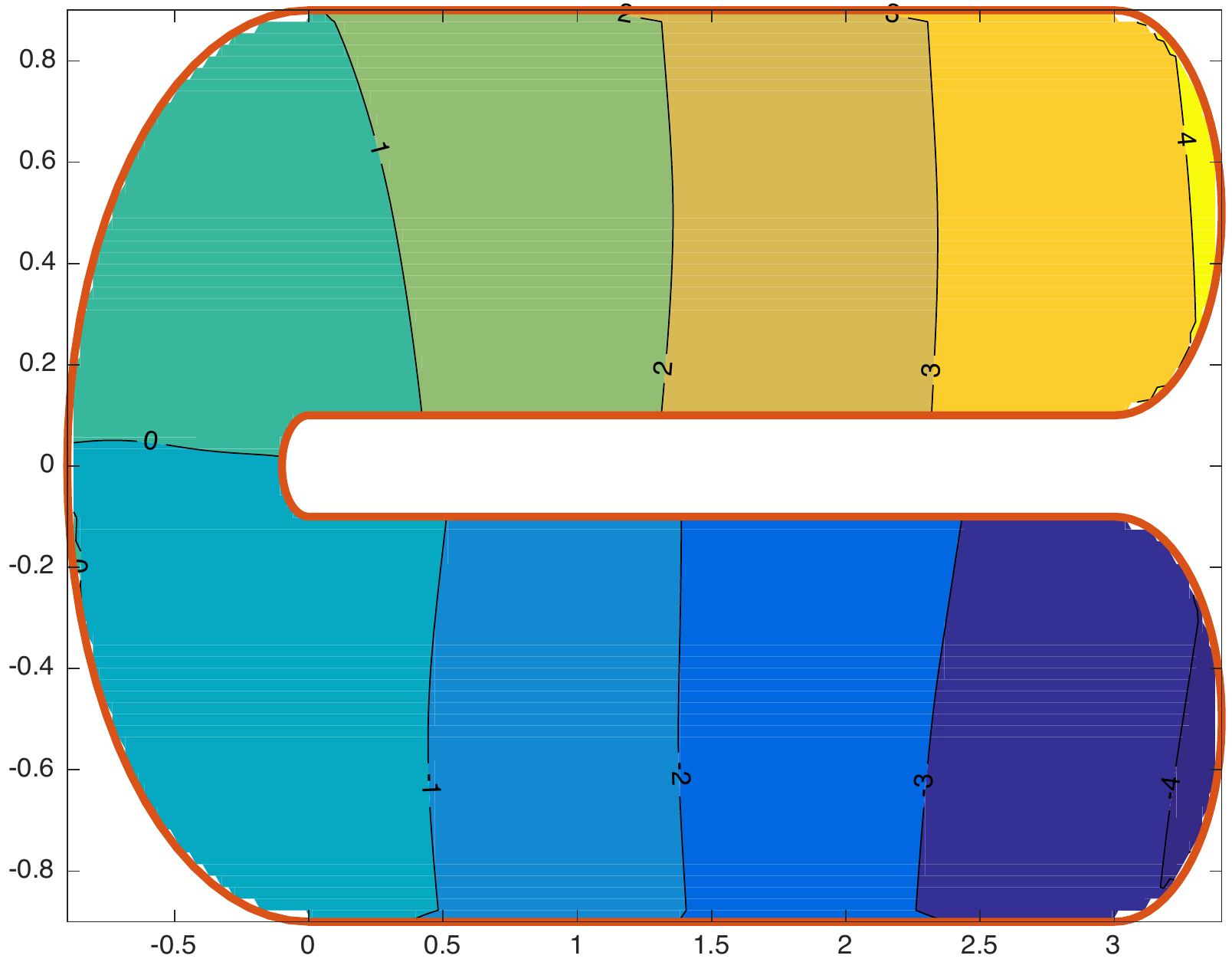} &\includegraphics[height=3cm]{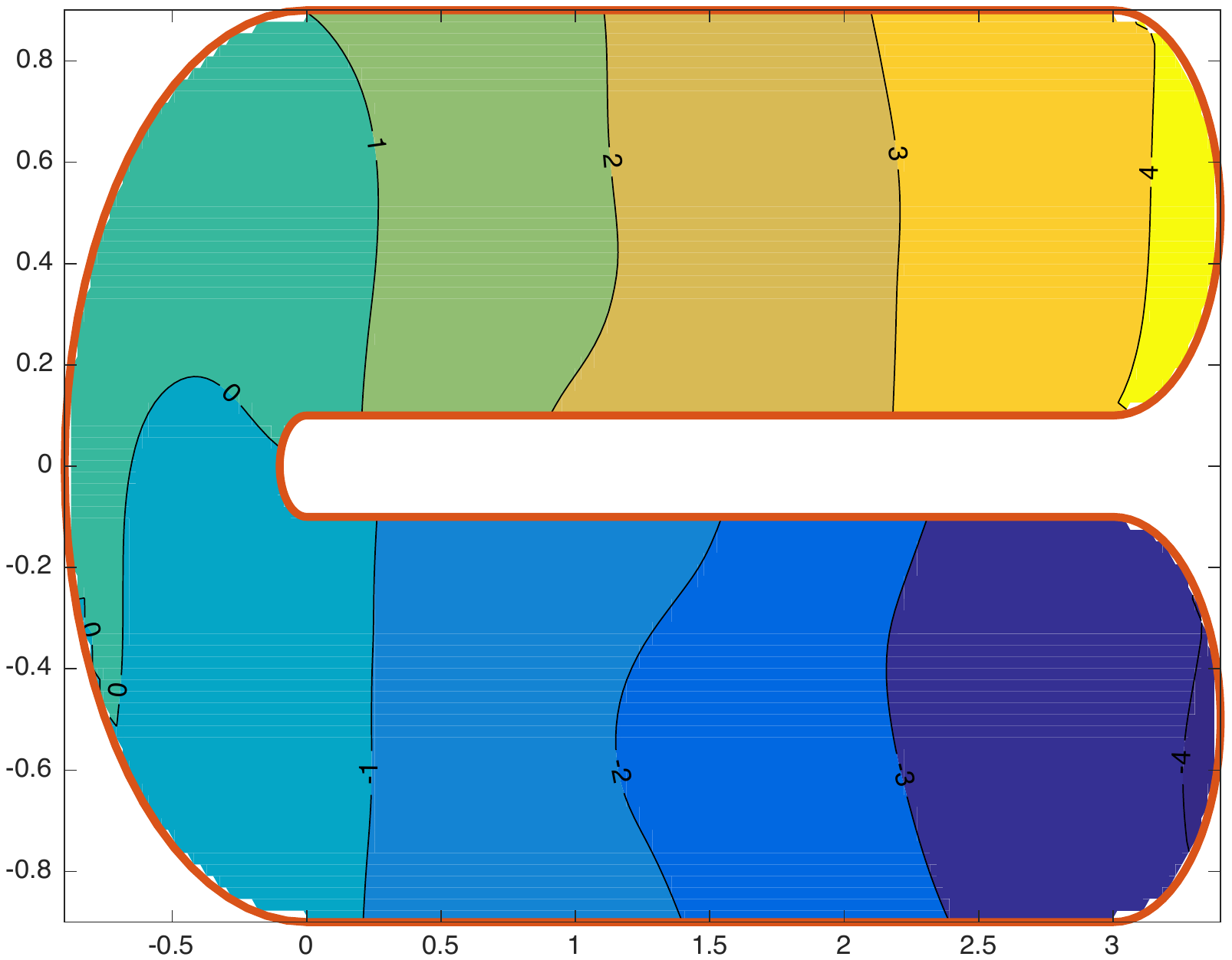} &\includegraphics[height=3cm]{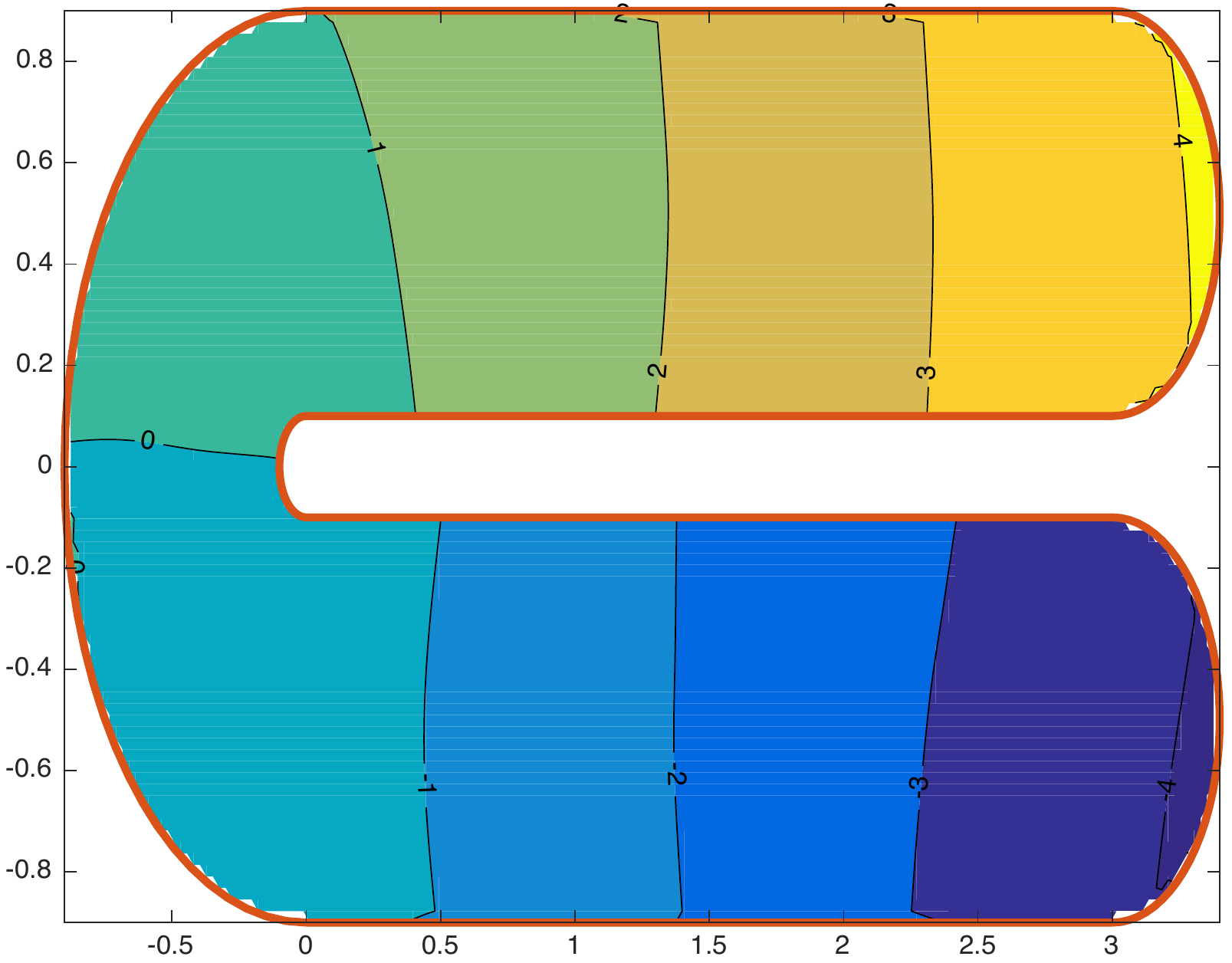}\\ [-6pt]
		\end{tabular}
	\end{center} \vskip -.15in
	\caption{Example 1. estimated functions using different triangulations when $n=200$.}
	\label{FIG:eg1-3}
\end{figure}	

Next we test the accuracy of the standard error estimation in (\ref{DEF:Sigma}) for $\widehat{\beta}_{1}$ and $\widehat{\beta}_{2}$. All the results based on triangulation $\triangle_2$ are listed in Table \ref{TAB:eg1-2}. The standard deviations of the estimated parameters computed based on 100 simulations are treated as the true standard errors (column labeled ``$\mathrm{SE}_{\mathrm{mc}}$''). Then we compared the mean and median of the 100 estimated standard errors calculated using (\ref{DEF:Sigma}) (columns labeled ``$\mathrm{SE}_{\mathrm{mean}}$'' and ``$\mathrm{SE}_{\mathrm{median}}$") with $\mathrm{SE}_{\mathrm{mc}}$. The column labeled ``$\mathrm{SE}_{\mathrm{mad}}$" is the interquartile range of the 100 estimated standard errors divided by 1.349. It can be  used as a robust estimate of the standard deviation. Table \ref{TAB:eg1-2} confirms the accuracy of the proposed standard error formula.

\begin{table}
	\begin{center}
		\caption{Example 1. standard error estimates of the coefficients using $\mathcal{S}$-PLSM-$\triangle_2$.}
		\label{TAB:eg1-2}
		\scalebox{0.85}{\begin{tabular*}{\columnwidth}{@{\extracolsep{\fill}}cccccccccc} \hline\hline
				\multirow{2}{*}{$\rho$} &\multicolumn{4}{c}{$\beta_1$}& &\multicolumn{4}{c}{$\beta_2$}\\ \cline{2-5} \cline{7-10}
				&$\mathrm{SE}_{\mathrm{mc}}$ &$\mathrm{SE}_{\mathrm{mean}}$ &$\mathrm{SE}_{\mathrm{median}}$ &$\mathrm{SE}_{\mathrm{mad}}$ &
				&$\mathrm{SE}_{\mathrm{mc}}$ &$\mathrm{SE}_{\mathrm{mean}}$ &$\mathrm{SE}_{\mathrm{median}}$ &$\mathrm{SE}_{\mathrm{mad}}$
				\\ \hline
				0.3 &0.0646 &0.0485 &0.0483 &0.0036 &&0.0264 &0.0243 &0.0242 &0.0013\\
				0.5 &0.0578 &0.0579 &0.0579 &0.0065 &&0.0286 &0.0243 &0.0243 &0.0015\\
				0.7 &0.0660 &0.0640 &0.0618 &0.0106 &&0.0273 &0.0243 &0.0243 &0.0015\\ \hline\hline
		\end{tabular*}}
	\end{center}
\end{table}

\vskip .12in \noindent \textbf{5.2. Example 2} \vskip .10in
\label{SUBSEC:example2}

In this example, we consider the case that the random noises are spatially correlated. Following \cite{Nandy:Lim:Maiti:17}, we consider a rectangle domain with 20$\times$20  lattice grid points, and then, for each of the 100 Monte Carlo experiments, we randomly sample $n=100$ grid points. The response variable $Y_{i}$'s are generated from the following model:
$Y_{i}=\mathbf{Z}_{i}^{\top}\bs{\beta}+\varepsilon_{i}$, $i=1,\ldots,n$,
where $\bs{\beta}=(1,-1,0,0,0,0,0,0)^{\top}$ and $\varepsilon$ is generated from a stationary gaussian process with mean zero. All the covariates are generated independently from $N(0,1)$.

We compare the selection and estimation performance of the $\mathcal{S}$-PLSM with the $\mathcal{S}$-SWR and the sparse linear model ($\mathcal{S}$-LM). 
For $\mathcal{S}$-SWR, we calculate the weight matrix using the gaussian covariance structure. The model selection and estimation results are summarized in Table \ref{TAB:eg2-1}. As expected, when the true error structure follows a stationary gaussian process, the $\mathcal{S}$-SWR performs perfect and the selection is $100\%$ correct. The linear model cannot capture the error structure in this scenario and it tends to increase false positive rate.
However, the proposed $\mathcal{S}$-PLSM method still performs really well in this case, and the correct selection rate achieves $98\%$, which demonstrates that our method is pretty robust in presence of spatial dependence.

\begin{table}
	\begin{center}
		\caption{Example 2. model selection and estimation results.}
		\label{TAB:eg2-1}
		\scalebox{0.85}{\begin{tabular*}{\columnwidth}{@{\extracolsep{\fill}}lccccc}\hline\hline
				\multirow{2}{*}{Method} &\multicolumn{3}{c}{Selection}  &\multicolumn{2}{c}{RMSE} \\ \cline{2-4} \cline{5-6}
				&F &T &C &$\beta_{1}$ &$\beta_{2}$ \\ \hline
				ORACLE &0.00 &6.00 &100 &0.0600 &0.0500 \\
				$\mathcal{S}$-LM &0.00 &5.81 & 89 &0.1230 &0.0884 \\
				$\mathcal{S}$-SWR &0.00 &6.00 &100 &0.0796 &0.0635 \\
				$\mathcal{S}$-PLSM &0.00 &5.98 &98 &0.0600 &0.0500 \\ \hline\hline
		\end{tabular*}}
	\end{center}
\end{table}

\vskip .12in \noindent \textbf{5.3. Example 3} \vskip .10in
\label{SUBSEC:example3}

We conduct another simulation study using the covariates and domain of the data from the mortality analysis described in Section 6. Specifically, we generate the response variable $Y_{i}$ from the following PLSM:
\begin{equation*}
Y_{i}=\mathbf{Z}_{i}^{\top}\bs{\beta}+\alpha(\mathbf{X}_{i})+\varepsilon_{i},~~~i=1,\ldots,n,
\end{equation*}
where $Z_{ij}$, $j=1,\ldots,11$, are the same as the explanatory variables used in the mortality data, the true $\beta_{j}$'s and $\alpha(\cdot)$ are set to be the same as the estimates obtained by PLSM with the SCAD penalty. The random error, $\varepsilon_{i},~i=1,\ldots,n$, are generated independently from $N(0,\widehat{\sigma}^{2})$ distribution, where $\widehat{\sigma}^{2}$ is the variance estimate of the measurement error obtained from the mortality data.

We fit an $\mathcal{S}$-PLSM and an $\mathcal{S}$-SWR with the SCAD penalty for the simulated dataset, where the triangulation used for the $\mathcal{S}$-PLSM is given in Figure \ref{FIG:Mortality-tri}. To see the effect of model misspecifiation on selection, we also consider a $\mathcal{S}$-LM with the SCAD penalty without including the spatial information. We repeat the generation and fitting procedures 100 times.

The variable selection and the parameter estimation results are summarized in Table \ref{TAB:eg3-1}. From this table, we find that the number of covariates  selected is much larger than the true number of nonzero components when the misspecified LM is used. The $\mathcal{S}$-SWR outperforms slightly the $\mathcal{S}$-LM in terms of the ``F" and ``T" values. However, the $\mathcal{S}$-PLSM has comparable performance with the ORACLE, and it performs much better than the $\mathcal{S}$-LM and the $\mathcal{S}$-SWR.

The last column in Table \ref{TAB:eg3-1} provides the 10-fold cross-validation RMSPE for the response variable, defined as $\left\{{n}^{-1}\sum_{m = 1}^{10} \sum_{i \in \kappa_{m}} (\widehat{Y}_i-Y_i) ^ 2\right\}^{1/2}$ over the 100 replications, where $\kappa_1, \ldots, \kappa_{10}$ comprise a random partition of the dataset into $10$ disjoint subsets of equal size. The cross-validation RMSPE shows the superior performance of the $\mathcal{S}$-PLSM as it provides more accurate predictions compared with the $\mathcal{S}$-LM though it includes fewer explanatory variables than the $\mathcal{S}$-LM.

\begin{table}
	\begin{center}
		\caption{Example 3. model selection and estimation results}
		\label{TAB:eg3-1}
		\scalebox{0.85}{\begin{tabular}{lcccccccccc}\hline\hline
				\multirow{3}{*}{Method}  &\multicolumn{3}{c}{Selection}
				&\multicolumn{6}{c}{RMSE}  & RMSPE\\ \cline{2-4}  \cline{6-10}
				&\multirow{2}{*}{F} &\multirow{2}{*}{T} &\multirow{2}{*}{C} & &Affluence &Disadvantage &ViolentCrime &Urban &\multirow{2}{*}{$\alpha(\cdot)$} &\multirow{2}{*}{$Y$}\\
				& & & & &$\beta_{4}$ &$\beta_{5}$ &$\beta_{6}$ &$\beta_{9}$ & &\\ \hline
				ORACLE &0.00 &7.00 &100 &&0.034 &0.020 &0.014 &0.013 &0.183 &0.766\\
				$\mathcal{S}$-LM &0.45 &3.03 &0 &&0.049 &0.091 &0.110 &0.080 &-- &0.860\\
				$\mathcal{S}$-SWR &0.08 &5.82 &60 &&0.025 &0.022 &0.031 &0.026 &-- &0.862\\
				$\mathcal{S}$-PLSM &0.06 &6.87 &86 &&0.034 &0.020 &0.021 &0.015 &0.184 &0.796\\
				\hline \hline
		\end{tabular}}
	\end{center}
	-- indicates the measurement is not applicable.
\end{table}

\setcounter{chapter}{6} \setcounter{equation}{0} \vskip .10in
\noindent \textbf{6. Application to U.S. Mortality Data} \label{SEC:application} \vskip 0.05in

We apply the proposed method to the United States mortality study. Mortality is an overall assessment of the population health of an area. The concentration of high mortality in specific areas in the U.S. has been an important public health concern and received considerable scholarly and policy attention in recent years \citep{Chen:Deng:12,Hoyert:12,Yang:Noah:15,Bauer:Kramer:16}. In the past few decades, the U.S. has witnessed an exceptional decrease in mortality, from almost 20 deaths per 1,000 population in 1930 to roughly 8 deaths per 1,000 population in 2010 \citep{Hoyert:12}. Despite the significant decrease in overall mortality through the years, disparities in mortality have persisted along various dimensions, such as, gender, age, race/ethnicity, income inequality, social affluence, concentrated disadvantage, safety and geographic space \citep{Chen:Deng:12,Yang:Noah:15}.

One of the goals of the study is to investigate the spatial pattern and identify important socioeconomic and racial/ethnic factors that affect mortality. The original mortality dataset is based on the county level, and it includes 3,037 counties from 48 states of the United States and the District of Columbia. The response variable is the average age-standardized mortality rates per 1,000 population based on county level over the period of 1998-2002, and it is publicly available from the Institute for Health Metrics and Evaluation \citep{ihme}. We classify all the counties in the dataset into six different groups according to their mortality rates: (i) less than 7.5, (ii) 7.50--9.00, (iii) 9.00--10.00, (iv) 10.00--11.00, (v) 11.00--12.50, and (vi) more than 12.50, and these groups are plotted in Figure \ref{FIG:Mortality}, which represents the observed mortality rate from each of 3037 counties in the United States.

Similar as in \cite{Chen:Deng:12,Sparks:Sparks:10,Yang:Jensen:11,Yang:Noah:15}, the explanatory variables in the study consist of many socioeconomic and racial/ethnic factors from year 2000, such as African-American rate, Hispanic/Latino rate, Gini coefficient, social affluence, disadvantage, violent crime rate (per 1,000 population), property crime rate (per 1,000 population), residential stability, urban rate, percentage of population without health insurance coverage and local government expenditure on health per population. Specifically, the information of Gini coefficient is publicly available at U.S. Census Bureau historical income tables (\url{https://www.census.gov/data/tables/time-series/dec/historical-income-counties.html}), and all the other explanatory variables can be obtained from U.S. Census Bureau and U.S. Federal Bureau Investigation (\url{https://www.census.gov/support/USACdataDownloads.html}). Meanwhile, the longitudes and latitudes of the geographic center of each county in the U.S. are available at \url{https://www.census.gov/geo/maps-data/data/gazetteer.html}.

According to \cite{Chen:Deng:12} and \cite{Yang:Noah:15}, social affluence is measured by the percentage of households that have income over \$75,000, the percentage of population obtaining at least a bachelor degree and percent of people in managerial and professional positions. As stated in \cite{Sparks:Sparks:10} median house value is another important socioeconomic factor that influences mortality rate. Therefore, we also include median house value as an indicator of social affluence. Based on \cite{Yang:Noah:15}, due to the highly positive correlation between those four variables, factor analysis is used to combine those four variables in a certain scale. Similarly, we apply factor analysis to combine public assistance rate, the percentage of female-headed families and the unemployment rate together to measure concentrated disadvantages. The factor of residential stability is measured by the percentage of population five years and over by residence in year 1995 lived in the same house in year 2000 and the ratio of housing units occupied by owners. As these two variables are highly correlated, following \cite{Yang:Noah:15}, we standardize each of them and take the average to get a single indicator for residential stability factor.

\begin{table}
	\begin{center}
		\caption{Variables in the mortality dataset}
		\label{TAB:covariates}
		\scalebox{0.85}{\begin{tabular}{ll}
				\hline \hline
				Variable &  Description \\
				\hline
				Mortality & mortality rate per 1,000 population \\
				AA$^{\ast}$ & African-American rate \\
				HL$^{\ast}$ & Hispanic/Latino rate \\
				Gini &  Gini coefficient showing the inequality between different levels of people in society\\
				Affluence$^{\ast}$ & social affluence factors:\\
				&percentage of households that have income over \$75,000\\
				&percentage of population obtaining at least a bachelor degree\\
				&percent of people in managerial and professional positions\\
				&median house value\\
				Disadvantage$^{\ast}$ & disadvantage factors:\\
				&public assistance rate\\
				&percentage of female-headed families\\
				&unemployment rate\\
				ViolentCrime$^{\ast}$ &  violent crime rate per 1000 population\\
				PropertyCrime$^{\ast}$ &  property crime rate per 1000 population\\
				ResidStab & residential stability \\
				Urban$^{\ast}$ & urban rate\\
				HealthCover$^{\ast}$ &  percentage of population without health insurance coverage\\
				ExpHealth$^{\ast}$ & local government expenditures in health per population\\
				Lat, Lon: & Latitude and longitude of the approximate geographic center of the county.\\
				\hline \hline
		\end{tabular}} 
	\end{center} \vspace{-.5cm}
	\begin{tablenotes}
		\small 
		\item Note: The covariates with $^{\ast}$ represent that they are transformed from the original value by $f(x)=\log (x+\delta)$. For example, $\text{AA}^{\ast}=\log(\text{AA}+\delta)$, where $\delta$ is a small number.
	\end{tablenotes}
\end{table}

As indicated in Table \ref{TAB:covariates}, we first apply the logarithm to each of the predictors except Gini coefficient and residential stability, then we standardize them before applying our method of variable selection. We fit the mortality data using the following PLSM:
\begin{align*}
\texttt{Mortality}&=\beta_0+\beta_1 \texttt{AA}+\beta_2 \texttt{HL}+\beta_3 \texttt{Gini}+\beta_4 \texttt{Affluence}+\beta_5 \texttt{Disadvantage} \notag\\
&~~+\beta_6 \texttt{ViolentCrime}+\beta_7 \texttt{PropertyCrime}+\beta_8 \texttt{ResidStab}\notag\\
&~~+\beta_9 \texttt{Urban}+\beta_{10} \texttt{HealthCover}+\beta_{11} \texttt{ExpHealth} + \alpha(\texttt{Lat},\texttt{Lon}).
\end{align*}

For the bivariate spline smoothing, we use a triangulation with 262 triangles and 167 vertices; see Figure \ref{FIG:Mortality-tri}). It has been proved in \cite{Lai:Schumaker:07}, when $d\geq 3r+2$, the bivariate spline achieves full approximation power, and thus, we suggest of using $d=5$ and $r=1$ when we generate the Bernstein basis polynomials. Then we apply the selection approach introduced in Section 2. Figure \ref{FIG:mortality-est} (d) plots the estimated surface of the $\alpha(\cdot)$ function in the PLSM.

\begin{figure}[htbp]
	\begin{center}
		\includegraphics[scale=0.2]{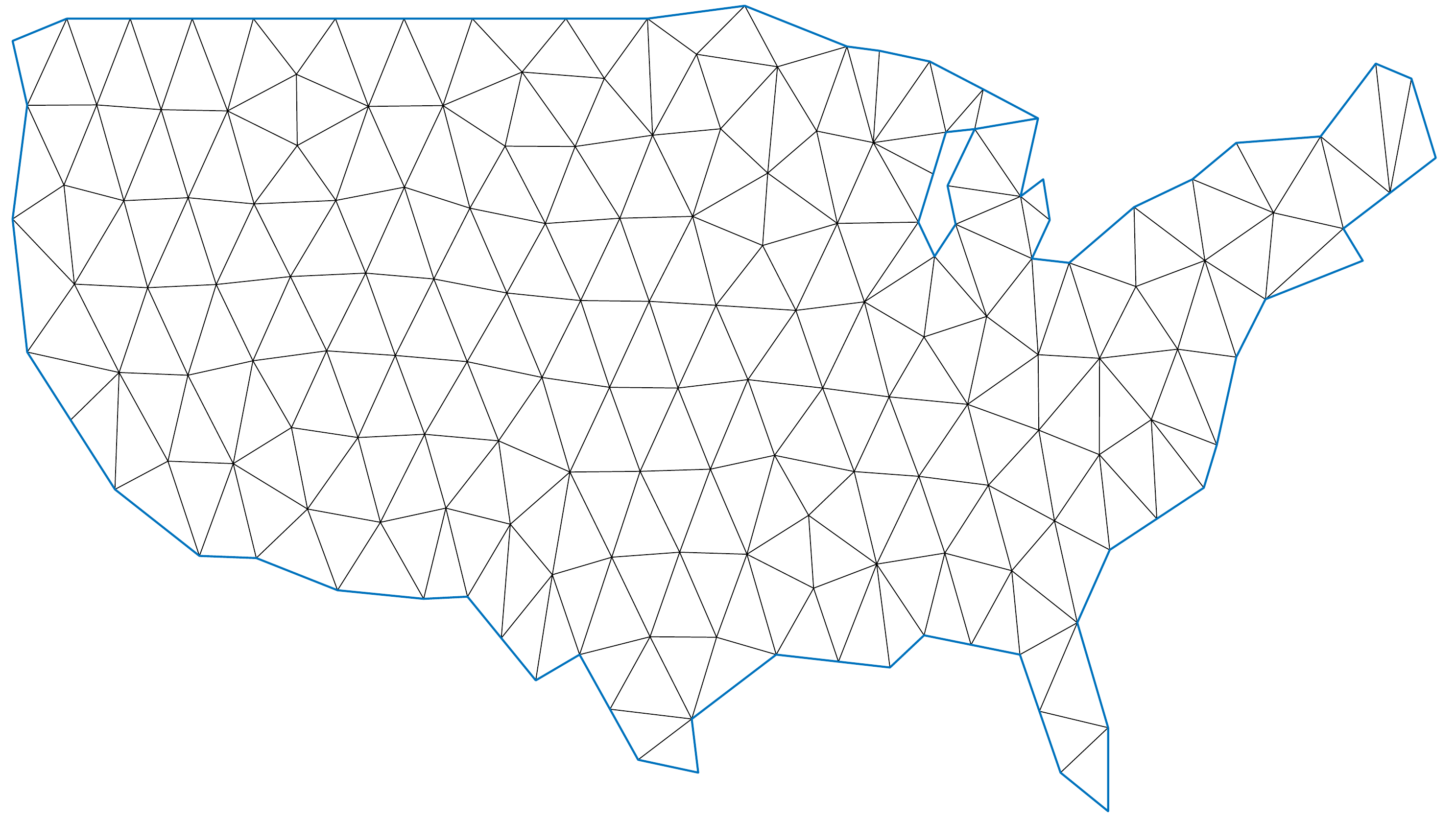}
		\caption{A triangulation of the domain of the U.S.}
		\label{FIG:Mortality-tri}
	\end{center}
\end{figure}

\begin{figure}[htbp]
	\begin{center}
		\begin{tabular}{cc}
			\includegraphics[scale=0.15]{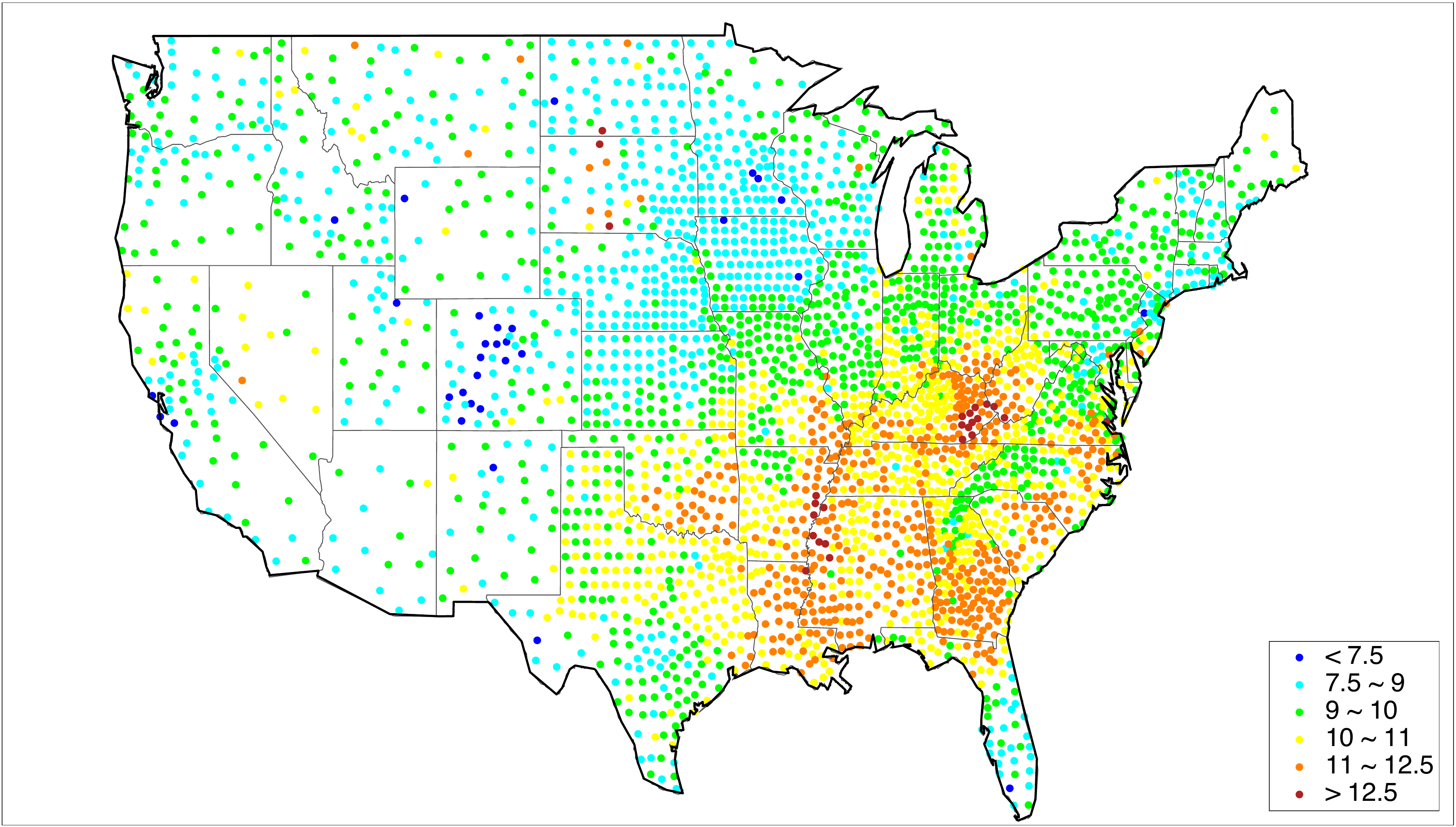} &\includegraphics[scale=0.15]{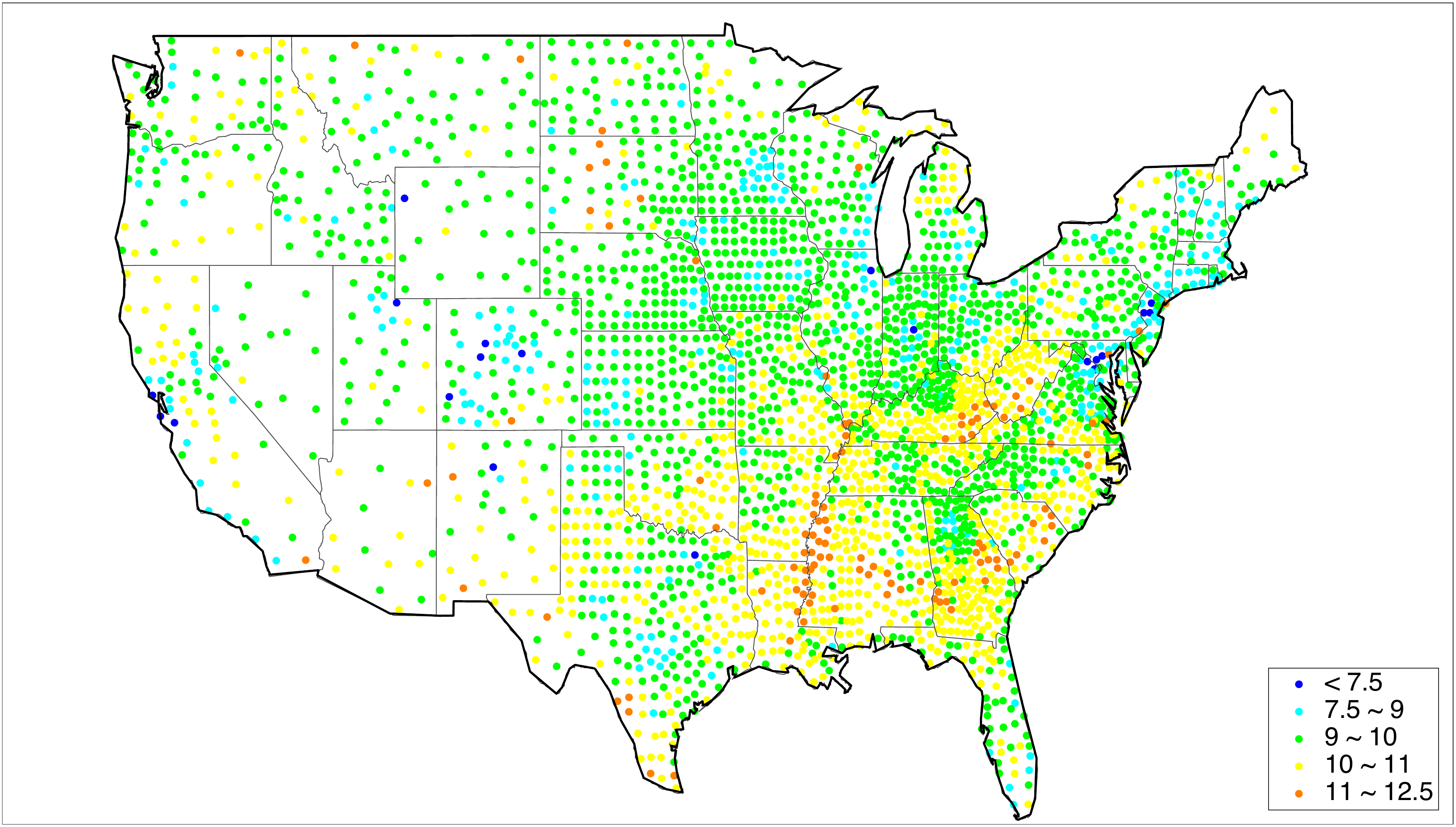}\\
			(a) &(b)\\[6pt]
			\includegraphics[scale=0.15]{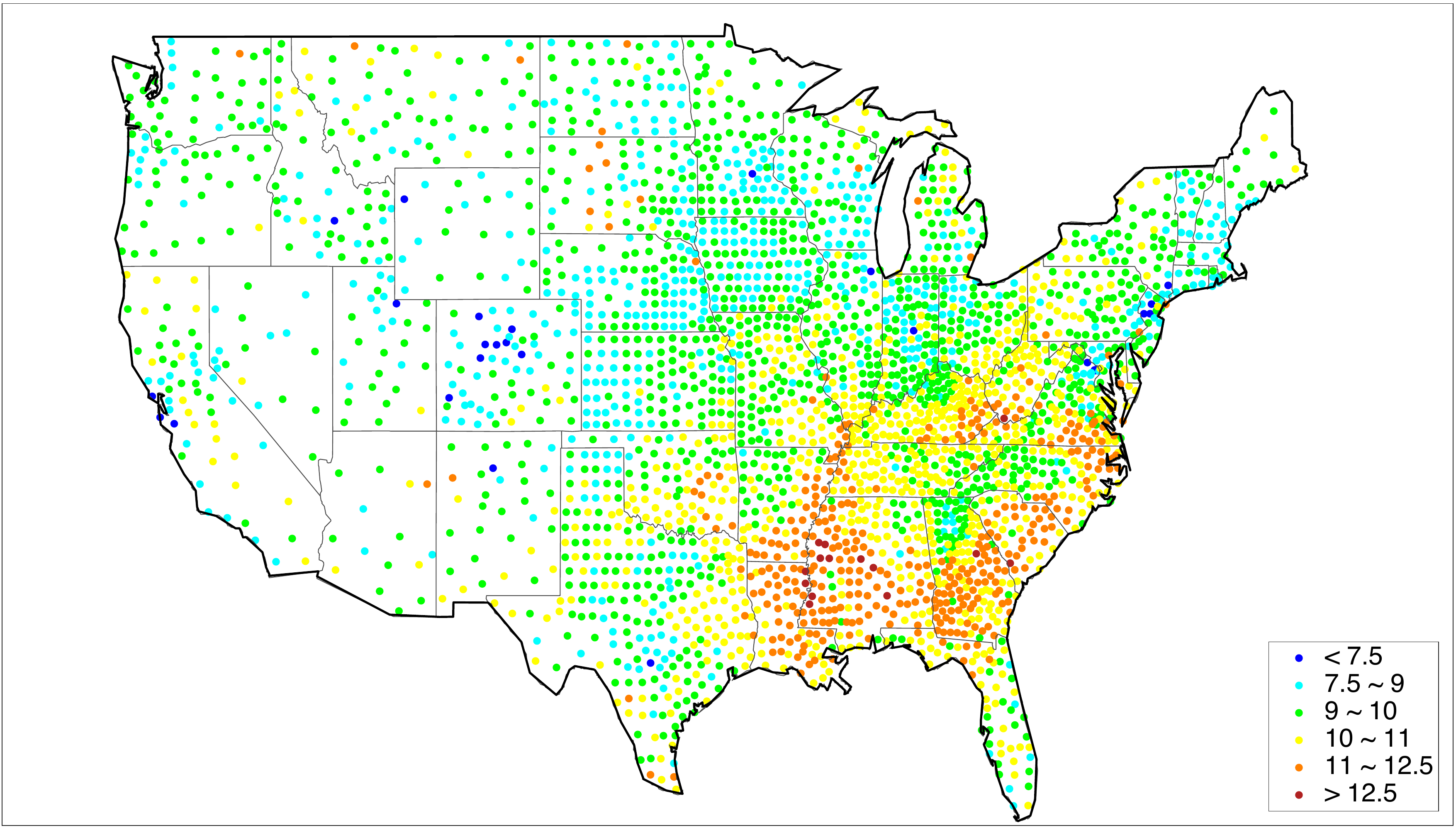} &~~~~~~\includegraphics[scale=0.2]{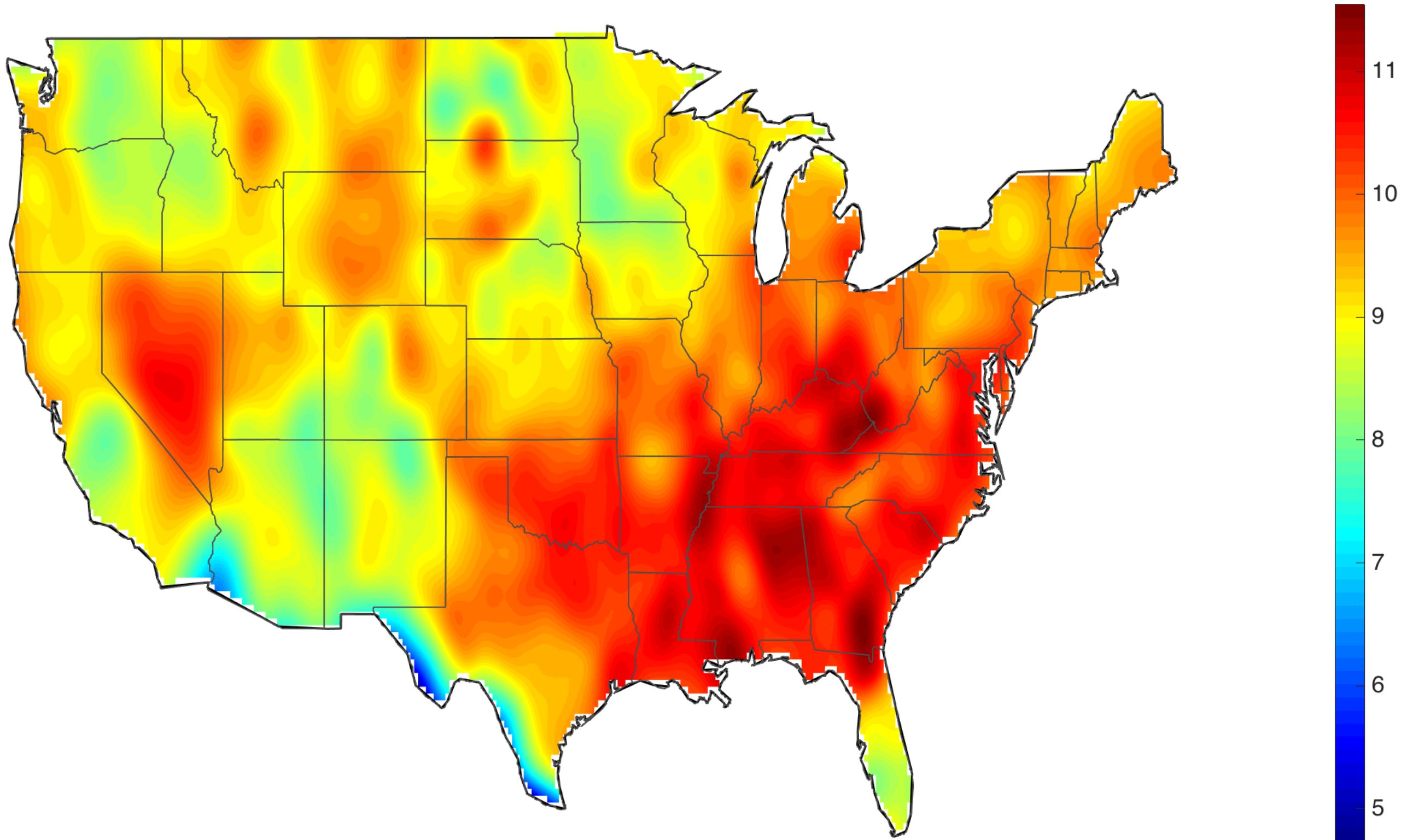}\\
			(c) &(d)\\
		\end{tabular}
	\end{center}
	\caption{(a) estimated mortality rate via the $\mathcal{S}$-PLSM with SCAD penalty; (b) estimated mortality rate via the $\mathcal{S}$-SWR with SCAD penalty; (c) estimated mortality rate via the $\mathcal{S}$-LM with SCAD penalty; (d) estimated spatial effect of $\alpha$ function via the $\mathcal{S}$-PLSM with SCAD penalty.}
	\label{FIG:mortality-est}
\end{figure}

The selected variables are presented in the second column in Table \ref{TAB:APPvs1}, from which one sees that $\mathcal{S}$-PLSM selects four explanatory  variables: Affluence, Disadvantage, ViolentCrime and Urban. The estimates of the coefficient (EST) and the standard errors (SE) for these selected variables with the associated $p$-values are shown in Columns 2--4 in Table \ref{TAB:APPvs1}. For comparison, we also analyze the data using the $\mathcal{S}$-SWR with a gaussian spatially weighted matrix and the na\"{i}ve $\mathcal{S}$-LM without adjusting the spatial correlation. Our method of variable selection has a strict sense of selecting variables in the sense of eliminating more variables. Table \ref{TAB:APPvs1} shows that our method provides a more parsimonious model and it eliminates four more variables among the variables selected by the $\mathcal{S}$-SWR or $\mathcal{S}$-LM. The results in Table \ref{TAB:APPvs1} also show that our method provides more accurate estimation with the mean squared error (MSE) of 0.2762, compared to the MSE of 0.8628 via $\mathcal{S}$-SWR and 0.6770 via $\mathcal{S}$-LM.

\begin{table}
	\begin{center}
		\caption{US morality rates: variable selection result.}
		\label{TAB:APPvs1}
		\scalebox{0.85}{\begin{tabular*}{\columnwidth}{@{\extracolsep{\fill}}lrrccc} \hline\hline
				\multirow{2}{*}{Variable} & \multicolumn{3}{c}{$\mathcal{S}$-PLSM} & \multirow{2}{*}{$\mathcal{S}$-SWR} &\multirow{2}{*}{$\mathcal{S}$-LM}\\ \cline{2-4}
				&EST &SE &$p$-value & \\ \hline
				AA & -- &-- &-- &\checkmark &-- \\
				HL & -- &-- &-- &\checkmark &-- \\
				Gini  &-- &-- &-- &-- &--   \\
				Affluence &$-$0.4666 &0.0160 &$<$0.0001 &\checkmark &\checkmark\\
				Disadvantage &0.4234 &0.0159 &$<$0.0001 &\checkmark &\checkmark\\
				ViolentCrime &0.0668 &0.0143 &$<$0.0001 &\checkmark &\checkmark\\
				PropertyCrime &-- &-- &--  &\checkmark &\checkmark\\
				ResidStab &-- &-- &-- &-- &--   \\
				Urban &0.1095 &0.0155 &$<$0.0001 &\checkmark &\checkmark\\
				HealthCover &-- &-- &-- &\checkmark&\checkmark\\
				ExpHealth &-- &-- &-- &\checkmark &--\\ \hline
				MSE &\multicolumn{3}{c}{0.2762} &0.8628 &0.6770\\
				MSPE &\multicolumn{3}{c}{0.4123} &0.8770 &0.6923\\ \hline\hline
		\end{tabular*}}
		
		{\footnotesize Note: ``\checkmark" indicates that variable is selected; ``--" indicates that variable is not selected.}
	\end{center}
\end{table} 

To further validate the variable selection and prediction results, we use 80\% of the observations to build the model and use the other 20\% to test the prediction accuracy. All the results are summarized based on 100 partitions. In a conclusion, we have African-American rate, social affluence, concentrated disadvantage, violent crime rate and urban rate as the selected significant variables.  Table \ref{TAB:APPvs1} shows that the mean squared prediction error (MSPE) of the mortality rate (per 1,000 population) is 0.6923 and 0.8770 for the $\mathcal{S}$-LM and $\mathcal{S}$-SWR, respectively, while the corresponding MSPE for the $\mathcal{S}$-PLSM is only 0.4123 with about $40\% \sim 50\%$ reduction.

We plot the estimated mortality rates via the $\mathcal{S}$-PLSM, the $\mathcal{S}$-SWR and the $\mathcal{S}$-LM with the SCAD penalty; see Figure \ref{FIG:mortality-est} (a)--(c), respectively. Both the $\mathcal{S}$-SWR and the $\mathcal{S}$-LM significantly underestimate the mortality rate in the South region of the U.S. and overestimate the mortality rate in the Midwest region. In contrast, the $\mathcal{S}$-PLSM fitting provides much more accurate estimates of the mortality rate.

Finally we perform model diagnostics for the $\mathcal{S}$-PLSM to check whether it adequately fits the data. Figure \ref{FIG:APP_res} (a) and (b) show a scatter plot and a histogram of the residuals of U.S. mortality rates. In addition, we conduct the Moran's I to test the spatial autoregression for the residuals. The test statistic is $-0.035$, and the $p$-value for the Moran's I test is $1$, which indicates that the process of the residuals is very likely a spatially independent random process.

\begin{figure}[htbp]
	\begin{center}
		\begin{tabular}{cc}
			\includegraphics[scale=0.21]{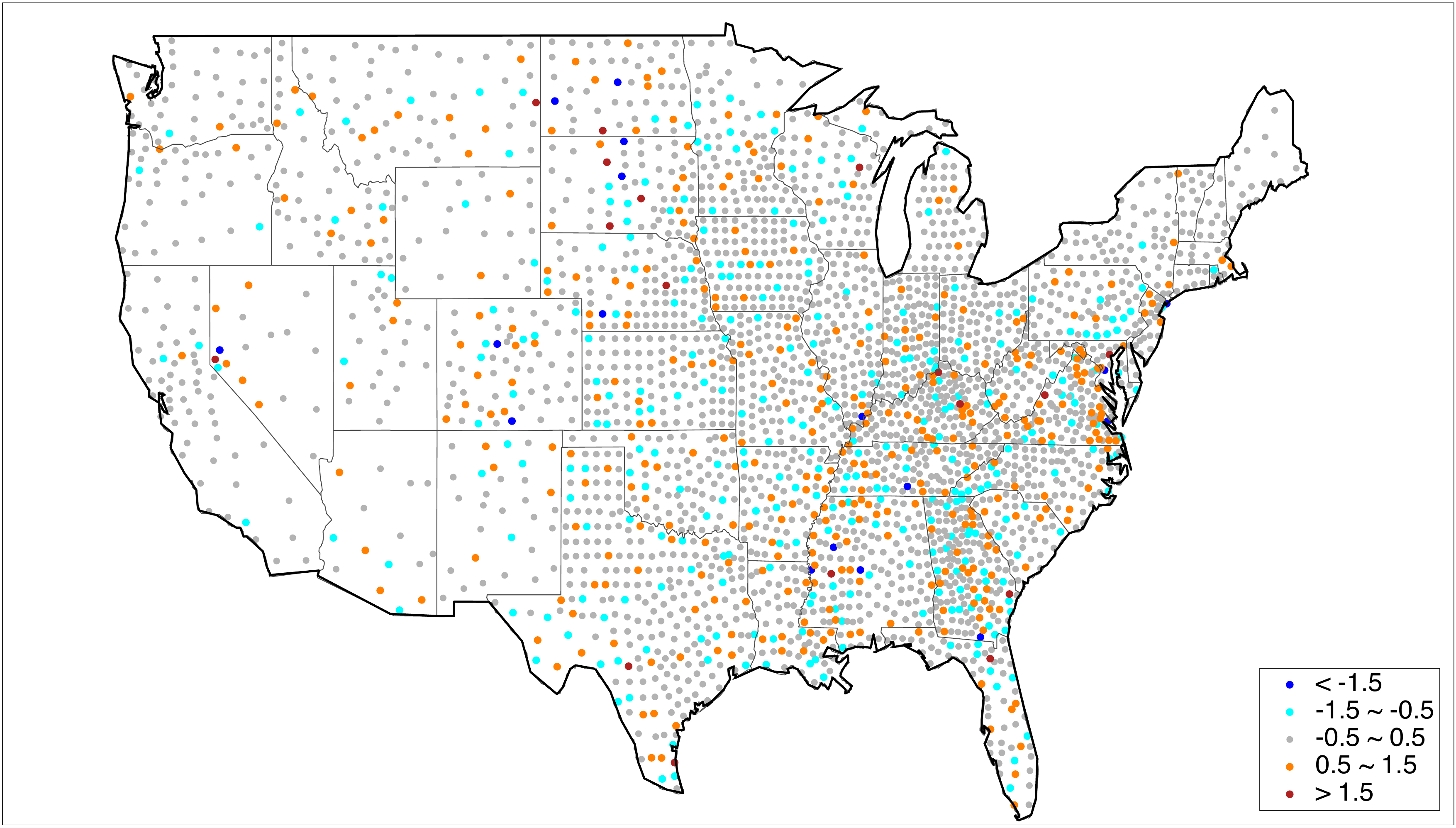} &\includegraphics[scale=0.2]{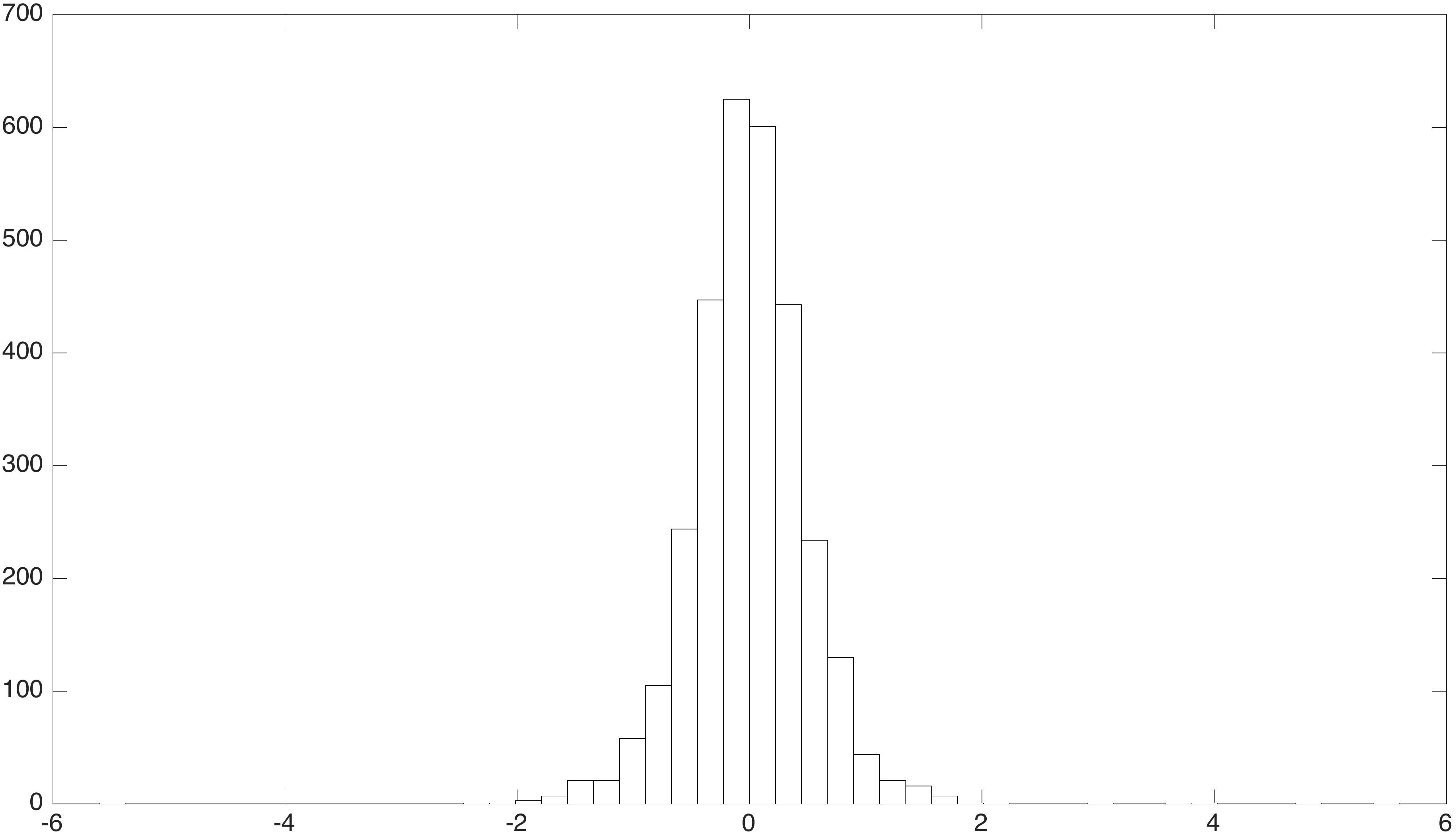} \\
			(a) &(b)\\
		\end{tabular}
	\end{center}
	\caption{(a) scatter plot and (b) histogram of the residuals of mortality rates via the $\mathcal{S}$-PLSM.}
	\label{FIG:APP_res}
\end{figure}

\setcounter{chapter}{7} \setcounter{equation}{0} \vskip .10in
\noindent \textbf{7. Concluding Remarks} \label{SEC:conclusion} \vskip 0.05in

In this study, we propose an efficient method for simultaneous estimation and variable selection in the PLSM for spatial data distributed on complex domains. When data are collected from irregularly shaped regions, we find in simulation studies that variable selection methods developed for regression models might usually perform poorly when the spatial information is ignored or handled inappropriately. This has motivated us for developing the proposed method in this paper. We adopt bivariate splines over triangulation to avoid the ``leakage" problem in the estimation of the nonparametric spatial component. A new type of double-penalized least squares has been developed to identify and estimate the components in the PLSM simultaneously, which is sufficiently fast for the user to analyze large data sets within seconds. The ``oracle" property of the proposed estimator of the parametric part has been established, and consistency of the proposed estimator of the nonparametric part is shown. The numerical results in the simulation demonstrate much better finite sample properties of the proposed estimators compared to the regression models when the spatial effect is unadjusted or adjusted inappropriately.

The selection consistency and the ``oracle" property obtained in this paper are based on the assumption that the errors are independent. Although this assumption is not uncommon in the nonparametric spatial smoothing literature, it is more realistic to relax the independence assumption. For example, \cite{Gao:Lu:Tjstheim:06} investigated the semiparametric spatial regression model for regularly spaced grid points under some stationary and mixing conditions. However, the data collected in our study are randomly distributed over complex domains with irregular boundaries. It is challenging to define the ``mixing" condition appropriately in this case, which warrants further research. As illustrated in Example 2 in the simulation studies, the spatial dependence can be alleviated by choosing an appropriate triangulation; it may not fully vanish, and certainly, there is more future work ahead to investigate this issue.

The proposed method in this paper can be easily extended to the case that $p$ is diverging or $p\gg n$, and our simulation studies have shown that the variable selection method also performs well for those cases. In future research, we will investigate the properties and performance of the proposed method for the more challenging high/ultra-high situation.

\setcounter{chapter}{8} \setcounter{equation}{0} \vskip .10in
\noindent \textbf{Acknowledgment}

Guannan Wang's research was partially supported by the Faculty Summer Research Grant from College of William \& Mary. The authors are very grateful to Ming-Jun Lai for providing us with the Matlab code on triangulation and bivariate spline basis construction. The authors would like to thank Lily Wang and Lei Gao for providing expertise that greatly assisted the research. The authors would like to thank the Editor, the Associate Editor and the referees for their constructive comments and suggestions.

\setcounter{chapter}{9} \setcounter{equation}{0} \vskip .10in
\noindent \textbf{Data Availability Statement}

The datasets that support the findings of this study are openly available. The response variable is the average age-standardized mortality rates per 1,000 population based on county level over the period of 1998-2002, and it is publicly available from the Institute for Health Metrics and Evaluation \citep{ihme}. The explanatory variables in the study consist of many socioeconomic and racial/ethnic factors from year 2000, such as African-American rate, Hispanic/Latino rate, Gini coefficient, social affluence, disadvantage, violent crime rate (per 1,000 population), property crime rate (per 1,000 population), residential stability, urban rate, percentage of population without health insurance coverage and local government expenditure on health per population. Specifically, the information of Gini coefficient is publicly available at U.S. Census Bureau historical income tables (\url{https://www.census.gov/data/tables/time-series/dec/historical-income-counties.html}), and all the other explanatory variables can be obtained from U.S. Census Bureau and U.S. Federal Bureau Investigation (\url{https://www.census.gov/support/USACdataDownloads.html}). Meanwhile, the longitudes and latitudes of the geographic center of each county in the U.S. are available at \url{https://www.census.gov/geo/maps-data/data/gazetteer.html.}

\newpage
\vskip 0.10in \noindent \textbf{Appendices}

\setcounter{chapter}{10} \renewcommand{\thetheorem}{A.\arabic{theorem}}
\renewcommand{\theproposition}{A.\arabic{proposition}}
\renewcommand{\thelemma}{A.\arabic{lemma}}
\renewcommand{\thecorollary}{A.\arabic{corollary}}
\renewcommand{\theequation}{A.\arabic{equation}} \renewcommand{\thesubsection}{A.\arabic{subsection}}
\renewcommand{\thetable}{{\arabic{table}}} \setcounter{table}{0}
\renewcommand{\thefigure}{\arabic{figure}} \setcounter{figure}{0}
\setcounter{equation}{0} \setcounter{lemma}{0} \setcounter{proposition}{0}
\setcounter{theorem}{0} \setcounter{subsection}{0}\setcounter{corollary}{0}
\vskip .05in \noindent \textbf{A. Some Preliminary Results}

For any function $f$ defined over domain $\Omega$,
let $E_{n}\left( f\right) =n^{-1}\sum_{i=1}^{n}f\left( \mathbf{X}_{i}\right) $ and $E\left(
f\right) =E[f\left( \mathbf{X}\right) ]$. Define the empirical inner product and norm as
$\left\langle f_{1},f_{2}\right\rangle _{n}=E_{n}\left( f_{1}f_{2}\right) $ and $\left\|
f_{1}\right\| _{n}^{2}=\left\langle f_{1},f_{1}\right\rangle _{n}$ for measurable
functions $f_{1}$ and $f_{2}$ on $\Omega$. The theoretical $L^2$ inner product and the induced
norm are given by $\left\langle f_{1},f_{2}\right\rangle _{L^2}=E\left(f_{1}f_{2}\right) $
and $\left\| f_{1}\right\|_{L^2}^{2}=\left\langle f_{1},f_{1}\right\rangle
_{L^2}$. Furthermore, let $\left\| \cdot\right\|
_{\mathcal{E}_{\upsilon }}$ be the norm introduced by the inner product $\left\langle \cdot, \cdot\right\rangle _{\mathcal{E}_{\upsilon
}}$, where, for $g_{1}$ and $g_{2}$ on $\Omega$,
\begin{equation*}
\left\langle g_{1},g_{2}\right\rangle_{\mathcal{E}_{\upsilon}} = \int_\Omega \sum_{i+j=\upsilon}
\binom{\upsilon}{i}
\left(\frac{\partial^{(\upsilon)}}{\partial x_1^i \partial x_2^j} g_1\right)
\left(\frac{\partial^{(\upsilon)}}{\partial x_1^i \partial x_2^j}g_2\right)dx_{1}dx_{2}.
\end{equation*}

We cite Lemma 2 in the Supplement of \cite{Lai:Wang:13} below, which shows that the uniform difference between the empirical and theoretical inner products is negligible.

\begin{lemma}
	\label{LEM:Rnorder}
	Let $f_{1}=\sum_{\xi \in \mathcal{K}}c_{\xi}B_{\xi }$, $f_{2}=\sum_{\zeta \in \mathcal{K}}\widetilde{c}_{\zeta}B_{\zeta }$ be any spline functions in $\mathbb{S}$. Under Assumption 7, we have
	\[
	\sup\limits_{f_{1},f_{2}\in \mathbb{S}}\left|
	\frac{\left\langle f_{1},f_{2}\right\rangle_{n}-\left\langle f_{1},f_{2}\right\rangle _{L^2}}{\left\|
		f_{1}\right\| _{L^2}\left\| f_{2}\right\| _{L^2}}\right| =O_{P}\left\{(N\log n)^{1/2}/{n}^{1/2}\right\}.
	\]
\end{lemma}

Following Lemma A.7 in \cite{Wang:Wang:Lai:Gao:18}, it is easy to obtain the following result in Lemma \ref{LEM:div2}.

\begin{lemma}
	Under Assumptions 1, 2, 7 and 8, there exist constants $0 < c_{Z} < C_{Z} < \infty$, such that with probability approaching 1 as $n\rightarrow \infty$, $c_{Z}\mathbf{I}_{p\times p}
	\leq n^{-1}(\mathbf{Z}-\widehat{\mathbf{Z}})^{\top}(\mathbf{Z}-\widehat{\mathbf{Z}}) \leq C_{Z}\mathbf{I}_{p\times p}$, where $\widehat{\mathbf{Z}}=\mathbf{H}_{\mathbf{B}}(\lambda_{1})\mathbf{Z}$ with  $\mathbf{H}_{\mathbf{B}}(\lambda_{1})$ in (\ref{DEF:HB}).
	\label{LEM:div2}
\end{lemma}

In the following, for any bivariate function $f(\cdot)$ and $\lambda>0$, define
\begin{equation*}
s_{\lambda,f}=\mathrm{argmin}_{s\in \mathbb{S}} \sum_{i=1}^{n}\{f(\mathbf{X}_{i})-s(\mathbf{X}_{i})\}^{2}+\lambda\mathcal{E}_{\upsilon}(s)
\end{equation*}
the penalized spline estimator of $f(\cdot)$. Then $s_{0,f}$ is the nonpenalized estimator of $f(\cdot)$.

Let $\nabla L(\bs{\beta})$ and $\nabla^{2} L(\bs{\beta})$ be the first order and second order partial derivatives of $L(\bs{\beta})$  in (6), then
$\nabla L\left(\bs{\beta}\right) =-(\mathbf{Z}-\widehat{\mathbf{Z}})^{\top}(\mathbf{Y}-\mathbf{Z}\bs{\beta})$ and
$\nabla^{2} L\left(\bs{\beta}\right) =(\mathbf{Z}-\widehat{\mathbf{Z}})^{\top}\mathbf{Z}$,
where
\begin{equation}
\widehat{\mathbf{Z}}=\mathbf{H}_{\mathbf{B}}(\lambda_{1})\mathbf{Z},
\label{DEF:Zhat}
\end{equation}
and according to the proof of Lemma A.10 in \cite{Wang:Wang:Lai:Gao:18}, $n^{-1}\nabla^{2} L\left(\bs{\beta}\right)=n^{-1}(\mathbf{Z}-\widehat{\mathbf{Z}})^{\top}(\mathbf{Z}-\widehat{\mathbf{Z}})+o_P(1)$.

\newpage
\setcounter{chapter}{9} \renewcommand{\thetheorem}{B.\arabic{theorem}}
\renewcommand{\theproposition}{B.\arabic{proposition}}
\renewcommand{\thelemma}{B.\arabic{lemma}}
\renewcommand{\thecorollary}{B.\arabic{corollary}}
\renewcommand{\theequation}{B.\arabic{equation}} \renewcommand{\thesubsection}{B.\arabic{subsection}}
\renewcommand{\thetable}{{B.\arabic{table}}} \setcounter{table}{0}
\renewcommand{\thefigure}{B.\arabic{figure}} \setcounter{figure}{0}
\setcounter{equation}{0} \setcounter{lemma}{0} \setcounter{proposition}{0}
\setcounter{theorem}{0} \setcounter{subsection}{0} \setcounter{corollary}{0}
\vskip .05in \noindent \textbf{B. Proof of Theorem 1}

Let $\tau _{n}=n^{-1/2}+a_{n,\lambda_{2}}$. It suffices to show that for any given $\zeta >0$, there exists a large constant $C$ such that
\begin{equation}
\mathrm{Pr}\left\{\sup_{\Vert \mathbf{u}\Vert =C}R(\bs{\beta}%
_{0}+\tau _{n}\mathbf{u})>R(\bs{\beta}_{0})\right\}\geq
1-\zeta .  \label{EQ:existence}
\end{equation}%
Let $U_{n,1}=L(\bs{\beta}_{0}+\tau _{n}%
\mathbf{u})-L(\bs{\beta}_{0})$ and $U_{n,2}=n\sum_{k=1}^{q}\{p_{%
	\lambda _{2}}(|\beta _{k0}+\tau _{n}u_{k}|)-p_{\lambda _{2}}(|\beta
_{k0}|)\} $, where $q$ is the number of components of $\bs{\beta}_{10}$%
. Note that $p_{\lambda _{2}}\left( 0\right) =0$ and $p_{\lambda _{2}}\left(
|\beta |\right) \geq 0$ for all $\beta $. Thus, $R(\bs{\beta}_{0}+\tau _{n}\mathbf{u})-R(\bs{\beta}_{0})\geq
U_{n,1}+U_{n,2}$.

For $U_{n,1}$, we have $L(\bs{\beta}_{0}+\tau _{n}\mathbf{u})=L(\bs{\beta}_{0})+\tau _{n}\mathbf{u}^{\top}\nabla{L}(\bs{\beta}_{0})+\frac{1}{2}%
\tau _{n}^{2}\mathbf{u}^{\top}\nabla^{2}{L}(\bs{\beta}^{\ast})\mathbf{u}, $
where $\bs{\beta}^{\ast }=t( \bs{\beta}_{0}+\tau_n\mathbf{u}) +\left(
1-t\right) \bs{\beta}_{0}$, $t\in \lbrack 0,1]$, and $\nabla^{2}{L}(\bs{\beta}_{0})=(\mathbf{Z}-\widehat{\mathbf{Z}})^{\top}\mathbf{Z}$ with $\widehat{\mathbf{Z}}$ defined in (\ref{DEF:Zhat}). Let $\bs{\alpha}_{0}=(\alpha_{0}(\mathbf{X}_{1}),\ldots, \alpha_{0}(\mathbf{X}_{n}))^{\top}$.  Note that $-\nabla{L}(\bs{\beta}_{0})$ is equal to
\[
(\mathbf{Z}-\widehat{\mathbf{Z}})^{\top}(\mathbf{Y}-\mathbf{Z}\bs{\beta}_{0}) =(\mathbf{Z}-\widehat{\mathbf{Z}})^{\top}(\bs{\alpha}_{0}+\bs{\epsilon})
=\mathbf{Z}^{\top}\{\mathbf{I}-\mathbf{H}_{\mathbf{B}}(\lambda_{1})\}\bs{\alpha}_{0}
+\mathbf{Z}^{\top}\{\mathbf{I}-\mathbf{H}_{\mathbf{B}}(\lambda_{1})\}\bs{\epsilon}.
\]
Denote $\mathbf{Z}_{j}^{\top}=(Z_{1j},...,Z_{nj})$, and let $W_{j}=n^{-1}\mathbf{Z}_{j}^{\top}\{\mathbf{I}-\mathbf{H}_{\mathbf{B}}(\lambda_{1})\}\bs{\alpha}_{0}$, then, similar to the proof of Lemma A.7 in \cite{Wang:Wang:Lai:Gao:18}, we can decompose $W_j$ as follows:
\[
W_{j}=\langle z_{j}-h_{j}, \alpha_{0} -
s_{\lambda_1,\alpha_0} \rangle_{n}+\langle h_{j}-%
\widetilde{h}_{j}, \alpha_{0} -s_{\lambda_1,\alpha_0} \rangle_{n}
+\frac{\lambda_{1}}{n}\langle s_{\lambda_1,\alpha_0}, \widetilde{h}_{j} \rangle_{\mathcal{E}_{\upsilon}}= W_{j,1}+W_{j,2}+W_{j,3},
\]
where $h_{j}(\cdot)$ is defined in (11), and  $\widetilde{h}_{j}\in \mathbb{S}$ satisfy
\begin{equation}
\|\widetilde{h}_{j}-h_{j}\|_{\infty}\leq C \left| \triangle \right| ^{\ell +1}\left|h_{j}\right|_{\ell+1,\infty}.
\label{EQ:h_j_tilde}
\end{equation}
By Proposition 1 in \cite{Lai:Wang:13}, one has
\[
\left\|\alpha_{0}-s_{\lambda_1,\alpha_0}\right\| _{\infty}=O_{P}\left\{\left| \triangle \right| ^{\ell +1}\left|\alpha_{0}\right|
_{\ell+1,\infty}+\frac{\lambda_{1}}{n\left| \triangle \right| ^{3}}%
\left( \left|\alpha_{0}\right| _{2,\infty}+\left| \triangle \right| ^{\ell -1}\left|
\alpha_{0}\right| _{\ell +1,\infty}\right) \right\}.
\]
Next, note that $E\left(W_{j,1}\right) =0$, and
\[
\textrm{Var}\left( W_{j,1}\right) =\frac{1}{n^{2}}\sum\limits_{i=1}^{n}E\left[
\left\{Z_{ij} -h_{j}(\mathbf{X}_{i})\right\}\left(\alpha_{0}-s_{\lambda_1,\alpha_0}\right) \right] ^{2}
\leq \frac{\|\alpha_{0}-s_{\lambda_1,\alpha_0}\|_{\infty}^2 }{n}
\left\Vert z_{j}-h_{j}\right\Vert _{L^{2}}^{2},
\]
so one has
\begin{equation}
\left\vert W_{j,1}\right\vert =O_{P}\left\{\frac{\left| \triangle \right| ^{\ell +1}}{%
	n^{1/2}}\left|\alpha_{0}\right|
_{\ell+1,\infty}+\frac{\lambda_{1}}{n^{3/2}\left| \triangle \right| ^{3}}%
\left( \left|\alpha_{0}\right| _{2,\infty}+\left| \triangle \right| ^{\ell -1}\left|
\alpha_{0}\right| _{\ell +1,\infty}\right) \right\}.
\label{EQ:W1}
\end{equation}
For $W_{j,2}$, one has
\begin{align}
|W_{j,2}| \leq& \| h_{j}-\widetilde{h}_{j}\|_{n}\left\Vert \alpha_{0} -s_{\lambda_1,\alpha_0}
\right\Vert _{n} =O_{P}\left(\left| \triangle \right| ^{\ell +1}\left|h_{j}\right|
_{\ell+1,\infty}  \right)\label{EQ:W2}\\
&\times O_{P}\left\{\left| \triangle \right| ^{\ell +1}\left|\alpha_{0}\right|
_{\ell+1,\infty}+\frac{\lambda_{1}}{n\left| \triangle \right| ^{2}}%
\left( \left|\alpha_{0}\right| _{2,\infty}+\left| \triangle \right| ^{\ell -1}\left|\alpha_{0}\right|
_{\ell +1,\infty}\right) \right\}.\notag
\end{align}
Finally, one has
\begin{eqnarray}
\!\left\vert W_{j,3}\right\vert &\leq& \frac{\lambda_{1}}{n} \|s_{\lambda_1,\alpha_0}\|_{\mathcal{E}_{\upsilon}} \|\widetilde{h}_{j} \|_{\mathcal{E}_{\upsilon}}
\leq \frac{\lambda_{1}}{n} \|s_{0,\alpha_0}\|_{\mathcal{E}_{\upsilon}} \|\widetilde{h}_{j} \|_{\mathcal{E}_{\upsilon}} \label{EQ:W3} \\
&\leq& \frac{\lambda_{1}}{n} C_{1} \left( \left| \alpha_{0}\right| _{2,\infty}+
\left| \triangle \right| ^{\ell -1}\left| \alpha_{0}\right| _{\ell
	+1,\infty}\right) \left( \left| h_{j}\right| _{2,\infty}+
\left| \triangle \right| ^{\ell -1}\left| h_{j}\right| _{\ell
	+1,\infty}\right). \notag
\end{eqnarray}
Combining (\ref{EQ:W1})-(\ref{EQ:W3}), one has
\[
|W_{j}| =O_{P}\left[\frac{1}{\sqrt{n}}\left\{\left| \triangle \right| ^{\ell +1}\left|\alpha_{0}\right|
_{\ell+1,\infty}+\frac{\lambda_{1}}{n\left| \triangle \right| ^{3}}%
\left( \left|\alpha_{0}\right| _{2,\infty}+\left| \triangle \right| ^{\ell -1}\left|
\alpha_{0}\right| _{\ell +1,\infty}\right)\right\}
\right]
\]
for $j=1,\ldots,p$. Therefore, Assumptions 5--8 imply that
$\mathbf{Z}^{\top}\{\mathbf{I}-\mathbf{H}_{\mathbf{B}}(\lambda_{1})\}\bs{\alpha}_{0} =o_{P}(n^{1/2})$.

Next,
\begin{eqnarray*}
	\mathrm{Var} \left[\mathbf{Z}^{\top}\{\mathbf{I}-\mathbf{H}_{\mathbf{B}}(\lambda_{1})\}\bs{\epsilon}\left\vert
	\mathbf{Z},\mathbf{X}\right. \right]
	&=&\mathbf{Z}^{\top}
	\{\mathbf{I}-\mathbf{H}_{\mathbf{B}}(\lambda_{1})\}
	\{\mathbf{I}-\mathbf{H}_{\mathbf{B}}(\lambda_{1})\}
	\mathbf{Z}\sigma^{2}\notag\\
	&=&\sigma^2
	\sum_{i=1}^{n}
	(\mathbf{Z}_{i}-\widehat{\mathbf{Z}}_{i})(\mathbf{Z}_{i}-\widehat{\mathbf{Z}}_{i})^{\top},
	\label{EQ:var}
\end{eqnarray*}
where $\widehat{\mathbf{Z}}_{i}$ is the $i$th column of
$\mathbf{Z}^{\top} \mathbf{H}_{\mathbf{B}}(\lambda_{1})$.
Using Lemma \ref{LEM:div2}, we have
$\mathbf{Z}^{\top}\{\mathbf{I}-\mathbf{H}_{\mathbf{B}}(\lambda_{1})\}
\bs{\epsilon}=O_{P}(n^{1/2})$.
Thus, $\tau _{n}\mathbf{u}^{\top}\nabla{L}(\bs{\beta}_{0})=O_{P}(n^{1/2}\tau
_{n})\left\Vert \mathbf{u}\right\Vert $. Next according to the proof of Lemma A.10 in \cite{Wang:Wang:Lai:Gao:18} $n^{-1}\nabla^{2} L\left(\bs{\beta}\right) =n^{-1}(\mathbf{Z}-\widehat{\mathbf{Z}})^{\top}(\mathbf{Z}-\widehat{\mathbf{Z}})+o_P(1)
=E[(\mathbf{Z}_{i}-\widehat{\mathbf{Z}}_{i})(\mathbf{Z}_{i}-\widehat{\mathbf{Z}}_{i})^{\top}]+o_P(1)$, so one has $\frac{1}{2}\tau _{n}^{2}\mathbf{u}^{\top}\nabla^{2}{L}(\bs{\beta}_{0})\mathbf{u}=O_{P}(n\tau _{n}^{2})+o_{P}(1)$. Therefore,
\begin{equation}
U_{n,1}=O_{P}(n^{1/2}\tau _{n})+O_{P}(n\tau _{n}^{2})+o_{P}(1).
\label{EQ:Dn1}
\end{equation}%
For $U_{n,2}$, by a Taylor expansion
\begin{equation*}
p_{\lambda _{2}}(|\beta _{k0}+\tau _{n}u_{k}|)=p_{\lambda _{2}}(|\beta
_{k0}|)+\tau _{n}u_{k}p_{\lambda _{2}}^{\prime }\left( \left\vert \beta
_{k0}\right\vert \right) \mathrm{sgn}\left( \beta _{k0}\right) +\frac{1}{2}%
\tau _{n}^{2}u_{k}^{2}p_{\lambda _{2}}^{\prime \prime }\left( \left\vert
\beta _{k}^{\ast }\right\vert \right) ,
\end{equation*}%
where $\beta _{k}^{\ast }=(1-t)\beta _{k0}+t(\beta _{k0}+n^{-1/2}u_{k})$, $%
t\in [0,1]$, and
\[
p_{\lambda _{2}}(|\beta _{k0}+\tau_{n}u_{k}|)=p_{\lambda _{2}}(|\beta _{k0}|)+\tau_{n}u_{k}p_{\lambda _{2}}^{\prime }\left( \left\vert \beta_{k0}\right\vert \right) \mathrm{sgn}\left( \beta _{k0}\right) +\frac{1}{2}\tau _{n}^{2}u_{k}^{2}p_{\lambda _{2}}^{\prime
	\prime }\left( \left\vert \beta _{k0}\right\vert \right)+o(n^{-1}).
\]
Thus, by the Cauchy-Schwartz inequality,
\begin{eqnarray*}
	n^{-1}U_{n,2} &=&\tau _{n}\sum_{k=1}^{q}u_{k}p_{\lambda _{2}}^{\prime
	}\left( \left\vert \beta _{k0}\right\vert \right) \mathrm{sgn}\left( \beta
	_{k0}\right) +\frac{1}{2}\tau _{n}^{2}\sum_{k=1}^{q}u_{k}^{2}p_{\lambda _{2}}^{\prime \prime }\left( \left\vert \beta _{k0}\right\vert \right) \\
	&\leq &\sqrt{r}\tau _{n}a_{n,\lambda_2}\Vert \mathbf{u}\Vert +\frac{1}{2}\tau
	_{n}^{2}b_{n,\lambda_{2}}\Vert \mathbf{u}\Vert ^{2}=C\tau _{n}^{2}(\sqrt{q}+b_{n,\lambda_{2}}C).
\end{eqnarray*}%
As $b_{n,\lambda_{2}}\rightarrow 0$, the first two terms on the right hand side of (\ref{EQ:Dn1}) dominate $U_{n,2}$, by taking $C$ sufficiently large. Hence (\ref{EQ:existence}) holds for sufficiently large $C$.${\tiny \blacksquare }\medskip$

\setcounter{chapter}{10} \renewcommand{\thetheorem}{C.\arabic{theorem}}
\renewcommand{\theproposition}{C.\arabic{proposition}}
\renewcommand{\thelemma}{C.\arabic{lemma}}
\renewcommand{\thecorollary}{C.\arabic{corollary}}
\renewcommand{\theequation}{C.\arabic{equation}} \renewcommand{\thesubsection}{C.\arabic{subsection}}
\renewcommand{\thetable}{{C.\arabic{table}}} \setcounter{table}{0}
\renewcommand{\thefigure}{C.\arabic{figure}} \setcounter{figure}{0}
\setcounter{equation}{0} \setcounter{lemma}{0} \setcounter{proposition}{0}
\setcounter{theorem}{0} \setcounter{subsection}{0} \setcounter{corollary}{0}
\vskip .05in \noindent \textbf{Proof of Theorem \protect\ref{THM:ORACLE}}

We first show that the estimator $\widehat{\bs{\beta}}$ must possess the sparsity property $\widehat{\bs{\beta}}_{2}=0$, which is stated as follows.

\begin{lemma}
	\label{LEM:sparsity} Under the conditions of Theorem \ref{THM:ORACLE}, with
	probability tending to 1, for any given $\bs{\beta}_{1}$ satisfying
	that $\Vert \bs{\beta}_{1}-\bs{\beta}_{10}\Vert =O_{P}(n^{-1/2})$
	and any constant $C$,
	$R\{(\bs{\beta}_{1}^{\top},\mathbf{0}^{\top})^{\top}\}=\min_{\Vert \bs{\beta}_{2}\Vert \leq Cn^{-1/2}}R\{(\bs{\beta}_{1}^{\top},%
	\bs{\beta}_{2}^{\top})\}$.
\end{lemma}

\begin{proof}
	To prove that the minimizer is obtained at $\bs{\beta}_{2}=0$, it
	suffices to show that with probability tending to 1, as $n\rightarrow \infty
	$, for any $\bs{\beta}_{1}$ satisfying $\Vert \bs{\beta}_{1}-%
	\bs{\beta}_{10}\Vert =O_{P}(n^{-1/2})$, $\partial R(\bs{\beta}%
	)/\partial \beta _{k}$ and $\beta _{k}$ have different signs for $\beta
	_{k}\in (-Cn^{-1/2},Cn^{-1/2})$, for $k=q+1,\cdots ,p$. Note that
	\begin{equation*}
	\nabla{R}_{k}\left( \bs{\beta}\right) \equiv \frac{\partial
		R(\bs{\beta})}{\partial \beta _{k}}=\nabla{L}_{k}\left(
	\bs{\beta}\right) +np_{\lambda_{2}}^{\prime }\left( \left\vert
	\beta _{k}\right\vert \right) \mathrm{sgn}(\beta _{k}),
	\end{equation*}%
	where $\nabla{L}_{k}\left( \bs{\beta}\right) =\nabla{L}_{k}\left( \bs{\beta}_{0}\right) +\sum_{k^{\prime}=1}^{p}\nabla^{2}{L}_{kk^{\prime}}\{t%
	\beta_{k^{\prime}}+(1-t)\beta _{0k^{\prime}}\}\left( \beta
	_{k^{\prime}}-\beta _{0k^{\prime}}\right) $, $t\in [0,1]$. Let $\mathbf{e}_{k}$ be the zero vector except for an entry of one at position $k$, then
	\[
	\nabla{L}_{k}\left( \bs{\beta}_{0}\right)=-\mathbf{e}_{k}^{\top}\mathbf{Z}^{\top}(\mathbf{I}-\mathbf{H}_{\mathbf{B}}(\lambda_{1}))\bs{\alpha}_{0}
	-\mathbf{e}_{k}^{\top}\mathbf{Z}^{\top}(\mathbf{I}-\mathbf{H}_{\mathbf{B}}(\lambda_{1}))\bs{\epsilon}\\
	=-\mathbf{e}_{k}^{\top}\mathbf{Z}^{\top}(\mathbf{I}-\mathbf{H}_{\mathbf{B}}(\lambda_{1}))\bs{\epsilon}+o_{P}(n^{1/2}).
	\]
	According to Lemma A.10 in \cite{Wang:Wang:Lai:Gao:18},
	\begin{equation*}
	n^{-1}\nabla^{2}{L}\left( \bs{\beta}_{0}\right)=n^{-1}E\left\{(\mathbf{Z}_{i}-\widetilde{\mathbf{Z}}_{i})
	(\mathbf{Z}_{i}-\widetilde{\mathbf{Z}}_{i})^{\top}\right\}+o_{P}\left( 1\right),
	\end{equation*}%
	\begin{equation*}
	\frac{1}{n}\sum_{k^{\prime}=1}^{d_1}\nabla^{2}{L}_{kk^{\prime}}\left( \beta
	_{k^{\prime}}-\beta _{0k^{\prime}}\right) =(\bs{\beta}-\bs{\beta}%
	_{0})^{\top}\left[E\left\{(\mathbf{Z}_{i}-\widetilde{\mathbf{Z}}_{i})
	(\mathbf{Z}_{i}-\widetilde{\mathbf{Z}}_{i})^{\top}\right\}\mathbf{e}_{k}+o_{P}\left( 1\right) \right].
	\end{equation*}%
	Thus, for any $\bs{\beta}$ satisfying $\Vert \bs{\beta}-\bs{\beta}_{0}\Vert =O_{P}(n^{-1/2})$ as stated in the assumption, we have $%
	n^{-1}\nabla{L}_{k}\left( \bs{\beta}\right) =O_{P}(n^{-1/2})$. Therefore, for any nonzero $\beta _{k}$ and $k=q+1,\cdots,p$,
	\[
	\nabla{R}_{k}\left( \bs{\beta}\right)=n\lambda_{2}\left\{\lambda_{2}^{-1}p_{\lambda_{2}}^{\prime }\left(
	\left\vert \beta _{k}\right\vert \right) \mathrm{sgn}(\beta _{k})+O_{P}( n^{-1/2}\lambda_{2}^{-1})\right\}.
	\]
	Since $\liminf_{n\rightarrow \infty }\liminf_{\beta _{k}\rightarrow
		0^{+}}\lambda_{2}^{-1}p_{\lambda_{2}}^{\prime }(|\beta _{k}|)>0$ and $%
	\sqrt{n}\lambda_{2}\rightarrow \infty $, the sign of the derivative is
	determined by that of $\beta _{k}$. Thus, the desired result is obtained.
\end{proof}

\begin{proof}[Proof of Theorem \ref{THM:ORACLE}]
	From Lemma \ref{LEM:sparsity}, it follows that $\widehat{\bs{\beta}}%
	_{2}=\mathbf{0}$.
	\begin{eqnarray*}
		\nabla{R}\left( \bs{\beta}\right) &=&\nabla{L}(\bs{\beta}%
		_{0})+\nabla^{2}{L}(\bs{\beta}^{\ast })\left( \bs{\beta} -\bs{\beta}%
		_{0}\right) +n\left\{ p_{\lambda_{2}}^{\prime }\left( \left\vert \beta
		_{k0}\right\vert \right) \mathrm{sign}\left( \beta _{k0}\right) \right\}
		_{k=1}^{q} \\
		&&+ \sum_{k=1}^{q}\left\{p_{\lambda_{2}}^{\prime \prime }\left( \left\vert
		\beta _{k0}\right\vert \right) +o_{P}\left( 1\right) \right\} (		\widehat{\beta }_{k}-\beta_{k0}) ,
	\end{eqnarray*}%
	where $\bs{\beta}^{\ast }=t\bs{\beta}_{0}+\left( 1-t\right)
	\bs{\beta}$, $t\in \lbrack 0,1]$. Using an argument similar to the
	proof of Theorem \ref{THM:ROOTn}, it can be shown that there exists a $%
	\widehat{\bs{\beta}}_{1}$ in Theorem \ref{THM:ROOTn} that is a root-$n$
	consistent local minimizer of $R\left\{( \bs{\beta}_{1} ^{\top},%
	\mathbf{0}^{^{\top}}) ^{\top}\right\}$, satisfying $
	n^{-1}\nabla{R}
	\left\{(\widehat{\bs{\beta}}_{1}^{\top},%
	\mathbf{0}^{^{\top}})^{\top}\right\} =\mathbf{0}$.
	
	The left hand side of the above equation can be written as
	\[
	n^{-1}\mathbf{Z}_{1}^{\top}(\mathbf{I}-\mathbf{H}_{\mathbf{B}}(\lambda_{1}))\bs{\epsilon} +\left\{ p_{\lambda_{2}}^{\prime }\left(
	\left\vert \beta _{k0}\right\vert \right) \text{sign}\left( \beta
	_{k0}\right) \right\} _{k=1}^{q}+o_{P}( n^{-1/2})
	\]
	\[
	+\left[ E\left\{(\mathbf{Z}_{1i}-\widetilde{\mathbf{Z}}_{1i})(\mathbf{Z}_{1i}
	-\widetilde{\mathbf{Z}}_{1i})^{\top} \right\} +o_{P}\left( 1\right) \right](
	\widehat{\bs{\beta}}_{1}-\bs{\beta}_{10})
	+\left\{ \sum_{k=1}^{q}p_{\lambda_{2}}^{\prime \prime }\left( \left\vert
	\beta _{k0}\right\vert \right) +o_{P}\left( 1\right) \right\}(
	\widehat{\bs{\beta}}_{1}-\bs{\beta}_{10}) .
	\]
	Thus, one has
	\begin{eqnarray}
	\mathbf{0} &=&n^{-1}\mathbf{Z}_{1}^{\top}(\mathbf{I}-\mathbf{H}_{\mathbf{B}}(\lambda_{1}))\bs{\epsilon} +\mathbf{\kappa} _{n,\lambda_2}+o_{P}(n^{-1/2}) \nonumber \\
	&&+\left[E\left\{(\mathbf{Z}_{1i}-\widetilde{\mathbf{Z}}_{1i})(\mathbf{Z}_{1i}
	-\widetilde{\mathbf{Z}}_{1i})^{\top} \right\}+
	\mathbf{\Sigma}_{\lambda_2}+o_{P}(1) \right] (\widehat{\bs{\beta}}_{1}-\bs{\beta}_{10}).
	\label{EQ:d1}
	\end{eqnarray}
	Next we study the conditional variance of $\mathbf{Z}_{1}^{\top}(\mathbf{I}-\mathbf{H}_{\mathbf{B}}(\lambda_{1}))\bs{\epsilon}$ given $\mathbf{Z}_{1}$ and $\mathbf{X}$. We write
	\[	\mathrm{Var}\left\{\mathbf{Z}_{1}^{\top}(\mathbf{I}-\mathbf{H}_{\mathbf{B}}(\lambda_{1}))\bs{\epsilon}|\mathbf{Z}_{1}, \mathbf{X}\right\}=\sum_{i=1}^{n}(\mathbf{Z}_{1i}-\widehat{\mathbf{Z}}_{1i}) (\mathbf{Z}_{i}-\widehat{\mathbf{Z}}_{1i})^{\top}
	=\left(n\langle z_{j},z_{j^{\prime}}-s_{\lambda_{1},z_{j^{\prime}}}\rangle_{n}\right)_{1\leq j,j^{\prime}\leq q}.
	\]
	For $\widetilde{h}_{j}\in \mathbb{S}$ defined in (\ref{EQ:h_j_tilde}), one has
	\begin{equation}
	\langle z_{j},z_{j^{\prime}}-s_{\lambda_{1},z_{j^{\prime}}}\rangle_{n}
	=\langle z_{j}-\widetilde{h} _{j}, z_{j^{\prime}}-s_{\lambda_{1},z_{j^{\prime}}}\rangle_{n}
	+\frac{\lambda_{1}}{n}\langle s_{\lambda_{1},z_{j^{\prime}}}, \widetilde{h}_{j} \rangle_{\mathcal{E}_{\upsilon}}.
	\label{EQ:zj}
	\end{equation}
	Note that
	$\vert\langle s_{\lambda_{1},z_{j^{\prime}}}, \widetilde{h}_{j^{\prime}}\rangle_{\mathcal{E}_{\upsilon}}\vert
	\leq \|s_{\lambda_{1},z_{j^{\prime}}}\|_{\mathcal{E}_{\upsilon}} \|\widetilde{h}_{j^{\prime}} \|_{\mathcal{E}_{\upsilon}}
	\leq \|\widehat{z}_{j^{\prime},0}\|_{\mathcal{E}_{\upsilon}} \|\widetilde{h}_{j^{\prime}} \|_{\mathcal{E}_{\upsilon}}$, $
	\|s_{\lambda_{1},z_{j^{\prime}}}\|_{\mathcal{E}_{\upsilon}}\leq C|\triangle|^{-2} \|\widehat{z}_{j^{\prime},0}\|_{\infty}$. Thus,
	$\vert\langle s_{\lambda_{1},z_{j^{\prime}}}, \widetilde{h}_{j^{\prime}}\rangle_{\mathcal{E}_{\upsilon}}\vert
	\leq C |\triangle|^{-2} \|\widehat{z}_{j^{\prime},0}\|_{\infty} \|\widetilde{h}_{j^{\prime}} \|_{\mathcal{E}_{\upsilon}}
	\leq C^{*} |\triangle|^{-3} (| h_{j}^{\prime}|_{2,\infty}+
	|\triangle|^{\ell +1-\upsilon}| h_{j}^{\prime}|_{\ell+1,\infty})$. 
	We can decompose $\langle z_{j}-\widetilde{h} _{j}, z_{j^{\prime}}-s_{\lambda_{1},z_{j^{\prime}}}\rangle_{n}$ as follows:
	\begin{align}
	&\langle z_{j}-\widetilde{h} _{j}, z_{j^{\prime}}-s_{\lambda_{1},z_{j^{\prime}}}\rangle_{n}
	=\langle z_{j}-h_{j}, z_{j^{\prime}}-h_{j^{\prime}}\rangle_{n}
	+ \langle h_{j}-\widetilde{h}_{j}, h_{j^{\prime}}-\widetilde{h}_{j^{\prime}}\rangle_{n} +\langle z_{j}-h_{j}, h_{j^{\prime}}-\widetilde{h}_{j^{\prime}}\rangle_{n}\notag\\
	&
	+\langle h_{j}-\widetilde{h}_{j}, z_{j^{\prime}}-h_{j^{\prime}}\rangle_{n}
	+\langle z_{j}-h_{j}, \widetilde{h}_{j^{\prime}} -s_{\lambda_{1},z_{j^{\prime}}}\rangle_{n}
	+\langle h_{j}-\widetilde{h}_{j}, \widetilde{h}_{j^{\prime}}-s_{\lambda_{1},z_{j^{\prime}}} \rangle_{n}. \label{EQ:zj1}
	\end{align}
	According to (\ref{EQ:h_j_tilde}), the second term on the right side of (\ref{EQ:zj1}) satisfies that
	\begin{equation*}
	|\langle h_{j}-\widetilde{h}_{j}, h_{j^{\prime}}-\widetilde{h}_{j^{\prime}}\rangle _{\infty}|
	\leq \|h_{j}-\widetilde{h}_{j}\|_{\infty}
	\|h_{j^{\prime}}-\widetilde{h}_{j^{\prime}}\|_{\infty}
	=o_{P}(1).
	\label{EQ:zj12}
	\end{equation*}
	The third term on the right side of (\ref{EQ:zj1}) satisfies that
	\begin{equation*}
	\vert \langle z_{j}-h_{j}, h_{j^{\prime}}-\widetilde{h}_{j^{\prime}}\rangle  _{n}\vert
	\leq \left\{\| z_{j}-h_{j}\|_{L^{2}} (1+o_{P}(1))\right\}\|h_{j^{\prime}}-\widetilde{h}_{j^{\prime}}\|_{\infty}
	=o_{P}(1).
	\label{EQ:zj13}
	\end{equation*}
	Similarly, we have
	$|\langle h_{j}-\widetilde{h}_{j}, z_{j^{\prime}}-h_{j^{\prime}}\rangle_{n}|=o_{P}(1)$.
	From the triangle inequality, we have
	\[
	\|\widetilde{h}_{j} -s_{\lambda_{1},z_{j}} \| _{n} \leq
	\|\widetilde{h}_{j} -h_{j}\|_{n} + \|h_{j}-s_{0,z_{j}}\|_{n} + \|s_{0,z_{j}}-s_{\lambda_{1},z_{j}}\|_{n}.
	\]
	According to (\ref{EQ:h_j_tilde}) and Lemma A.9 in \cite{Wang:Wang:Lai:Gao:18}, 
	$\| \widetilde{h}_{j} -s_{\lambda_{1},z_{j}} \| _{n} \leq
	\|h_{j}-s_{0,z_{j}}\|_{n} + o_{P}(1)$.
	Let $h_{j,n}^{\ast}=\mathrm{argmin}_{h\in \mathbb{S}}\|z_j-h\|_{L^{2}}$, then, based on the triangle inequality, one has
	$\|h_{j}-s_{0,z_{j}}\|_{n}\leq\|h_{j}-h_{j,n}^{\ast }\|_{n}
	+\|h_{j,n}^{\ast}- s_{0,z_{j}}\|_{n}$. 
	It is clear that $\|h_{j}-h_{j,n}^{\ast}\|_{L^{2}}=o_{P}(1)$. By Lemma \ref{LEM:Rnorder}, one has $\|h_{j}-h_{j,n}^{\ast }\|_{n}=o_{P}(1)$.
	One also observes that $
	\|s_{0,z_{j}}-h_{j,n}^{\ast}\|_{L^{2}}^{2}=\|z_{j}-s_{0,z_{j}}\|_{L^{2}}^{2}-\|z_{j}- h_{j,n}^{\ast}\|_{L^{2}}^{2}$
	and
	$\|z_{j}-s_{0,z_{j}}\|_{n}\leq \|z_{j}- h_{j,n}^{\ast}\|_{n}$.
	Applying Lemma \ref{LEM:Rnorder} again, we have
	$\|s_{0,z_{j}}-h_{j,n}^{\ast}\|_{L^{2}}^{2}=o_{P}(\|z_{j}-h_{j,n}^{\ast}\|_{L^{2}}^{2})+o_{P}(\|z_{j}- s_{0,z_{j}}\|_{L^{2}}^{2})$. Moreover, there exists a constant $C$ such that $\|z_{j}- h_{j,n}^{\ast}\|_{L^{2}} \leq C$, and  
	$\|z_{j}-s_{0,z_{j}}\|_{L^{2}}\leq \|z_{j}- h_{j,n}^{\ast}\|_{L^{2}}+ \|h_{j,n}^{\ast}-s_{0,z_{j}}\|_{L^{2}}
	\leq C+ \|h_{j,n}^{\ast}-s_{0,z_{j}}\|_{L^{2}}$. Therefore, $\|h_{j,n}^{\ast}-s_{0,z_{j}}\|_{L^{2}}=o_{P}(1)$, then
	$\|h_{j,n}^{\ast }-s_{0,z_{j}}\|_{n}=o_{P}(1)$ by Lemma \ref{LEM:Rnorder}. Hence,
	\begin{equation}
	\|s_{0,z_{j}} -h_{j}\|_{n}=o_{P}(1).
	\label{EQ:zj-psijstar}
	\end{equation}
	Furthermore, by Lemma \ref{LEM:Rnorder} and (\ref{EQ:zj-psijstar}), one has
	\[
	\vert \langle z_{j}-h_{j}, \widetilde{h}_{j^{\prime}} -s_{\lambda_{1},z_{j^{\prime}}} \rangle_{n} \vert
	\leq  \left\{\| z_{j}-h_{j}\|_{L^{2}} (1+o_{P}(1))\right\}
	\left\{ \|h_{j}-s_{0,z_{j}}\|_{n} + o_{P}(1)\right\}=o_{P}(1).
	\]
	Similarly, one has
	\begin{equation}
	\vert \langle h_{j}-\widetilde{h}_{j}, \widetilde{h}_{j^{\prime}} -s_{\lambda_{1},z_{j^{\prime}}} \rangle_{n} \vert
	\leq  \|h_{j}-\widetilde{h}_{j}\|_{n}
	\left\{ \|h_{j}-s_{0,z_{j}}\|_{n} + o_{P}(1)\right\}=o_{P}(1).
	\label{EQ:zj6}
	\end{equation}
	Combining (\ref{EQ:zj})-(\ref{EQ:zj6}) yields
	$
	\langle z_{j},z_{j^{\prime}}-s_{\lambda_{1},z_{j^{\prime}}}\rangle_{n}
	= \langle z_{j}-h_{j},z_{j^{\prime}}-h_{j^{\prime}}^{\ast}\rangle_{n}  +o_{P}(1)
	$. Therefore,
	\begin{eqnarray*}
		n^{-1}\mathrm{Var}\left\{\mathbf{Z}_{1}^{\top}(\mathbf{I}-\mathbf{H}_{\mathbf{B}}(\lambda_{1}))\bs{\epsilon}|\mathbf{Z}_{1}, \mathbf{X}\right\}&=&n^{-1}\sum_{i=1}^{n}(\mathbf{Z}_{1i}-\widetilde{\mathbf{Z}}_{1i}) (\mathbf{Z}_{1i}-\widetilde{\mathbf{Z}}_{1i})^{\top}
		+o_{P}(1)\\
		&=&E[(\mathbf{Z}_{1i}-\widetilde{\mathbf{Z}}_{1i}) (\mathbf{Z}_{1i}-\widetilde{\mathbf{Z}}_{1i})^{\top}]+o_{P}(1),
	\end{eqnarray*}
	where $\widetilde{\mathbf{Z}}_{1i}=\left\{h_{1}(\mathbf{X}_{i}), \ldots, h_{q}(\mathbf{X}_{i})\right\}^{\top}$. By (\ref{EQ:d1}), Slutsky's Theorem and central limit theorem, one has
	$\sqrt{n}(\mathbf{\Sigma}_{s} +\mathbf{\Sigma} _{\lambda_{2} }) \left\{ \widehat{\bs{\beta}}_{1}-\bs{\beta}_{10}+(\mathbf{\Sigma}_{s} +\mathbf{\Sigma} _{\lambda_{2} }) ^{-1}\mathbf{\kappa}_{n,\lambda_{2}}\right\} \rightarrow \mathrm{N}(\mathbf{0},\sigma^{2}\mathbf{\Sigma}_{s})$ 
	using similar arguments as in the proof of Theorem 1 in \cite{Wang:Wang:Lai:Gao:18}, where
	$\mathbf{\Sigma}_{s}=\sigma^{-2}E[(\mathbf{Z}_{1}-\widetilde{\mathbf{Z}}_{1})
	(\mathbf{Z}_{1}-\widetilde{\mathbf{Z}}_{1})^{\top}]$.
	
	Hence the result in Theorem \ref{THM:ORACLE} is proved.
\end{proof}

\bibliographystyle{asa}
\bibliography{wileyNJD-APA}

\end{document}